\documentclass[11pt,draftcls,peerreviewca,onecolumn,a4paper,dvips]{IEEEtran}
\usepackage{amsfonts}
\usepackage{mathrsfs}
\usepackage{amsfonts}
\usepackage{amssymb}
\usepackage{graphicx}
\usepackage{amsmath}
\usepackage{array}
\usepackage{cases}

\usepackage{bbding}
\usepackage{subfigure}
\usepackage{latexsym}

\usepackage{cite,graphicx,amsmath,amssymb,color}
\usepackage{subfigure}

\newtheorem{Thm}{Theorem}
\newtheorem{Lem}{Lemma}
\newtheorem{Cor}{Corollary}
\newtheorem{Def}{Definition}

\newtheorem{Alg}{Algorithm}

\newtheorem{Prob}{Problem}
\newtheorem{Rem}{Remark}

\newtheorem{Asump}{Assumption}

\usepackage{multirow}

\begin{document}

\title{\huge{\textcolor{black}{Delay-aware BS Discontinuous Transmission Control and User Scheduling for Energy Harvesting Downlink Coordinated MIMO Systems}}}
\author{\thanks{This work was supported by Huawei Technologies Co. Ltd.}\authorblockN{Ying Cui, \IEEEmembership{MIEEE}, Vincent K.~N.~Lau,
\IEEEmembership{FIEEE}, Yueping Wu, \IEEEmembership{StMIEEE}}
\authorblockA{ECE Department, Hong Kong University of Science and Technology, Hong Kong\\
Email: cuiying@ust.hk, eeknlau@ece.ust.hk, eewyp@ust.hk} }

\maketitle

\begin{abstract}

In this paper, we propose a {\em two-timescale} delay-optimal 
base station Discontinuous Transmission (BS-DTX) control and user
scheduling for downlink coordinated MIMO systems with energy harvesting capability. \textcolor{black}{To reduce the complexity and signaling overhead in practical systems,} the BS-DTX control is adaptive to both the energy state
information (ESI) and the data queue state information (QSI) over a
longer \textcolor{black}{timescale}. The user scheduling is adaptive to the ESI,
the QSI and the channel state information (CSI) over a shorter
timescale. We show that the two-timescale delay-optimal control
problem can be modeled as an infinite horizon average cost Partially
Observed Markov Decision Problem (POMDP), which is well-known to be
a difficult problem  in general. By using sample-path analysis and exploiting specific problem
structure, we first obtain some structural
results on the optimal control policy and derive an {\em equivalent
Bellman equation} with reduced state space.  To reduce the
complexity and facilitate distributed implementation, we obtain a delay-aware
distributed solution with the BS-DTX control at the \textcolor{black}{BS controller} (BSC) and the user
scheduling at each cluster manager (CM) using approximate dynamic programming and
distributed stochastic learning.  We show that the proposed distributed
two-timescale algorithm converges almost surely. Furthermore, using queueing theory,
stochastic geometry and optimization techniques, we derive sufficient conditions
for the data queues to be stable in the \textcolor{black}{coordinated} MIMO network
and discuss various design insights. Finally, we compare the
proposed algorithm with various baseline schemes and show that
significant delay performance gain can be achieved.
\end{abstract}


\newpage

\section{Introduction}\label{sec:intro}

\textcolor{black}{Inter-cell interference is a critical performance bottleneck in cellular networks.
The interference mitigation techniques can be roughly classified into two types, namely coordinated MIMO techniques and cooperative MIMO techniques, according to the required backhaul consumption \cite{GesbertCoMIMOnewlook:2010}. For coordinated MIMO techniques, only the channel state information (CSI) is shared among MIMO base stations (BSs) through backhaul for the coordinated beamforming design at each BS to combat interference\cite{WeiYuCoordinatedBF:2010}. On the other hand, for cooperative MIMO techniques, both the CSI and the payload data are shared among MIMO BSs through backhaul for joint precoder designs at all the BSs  to combat interference\cite{RobertNetworkMIMO:09}. Since CSI sharing is performed for each transmission frame, while data sharing is operated for each data symbol, coordinated MIMO consumes much less backhaul capacity than cooperative MIMO at the expense of performance (e.g., degrees of freedom). }

\textcolor{black}{Due to the limited degrees of freedom and the limited backhaul capacity at each BS, global cooperation or coordination of all the BSs in the network is not possible and the BSs are organized into disjoint {\em clusters} \cite{KBCellularCoodinateMIMO:2011,RobertNetworkMIMO:09,Jun2007:networkmimo,Howard2007:networkmimo,Brueck2010:networkmimo}. The BSs within each cluster cooperatively serve the users associated with them, which lowers the system complexity and completely eliminates intra-cluster interference. For example,  in \cite{KBCellularCoodinateMIMO:2011},  multi-antenna BSs in each fixed cluster adopt coordinated beamforming to serve the single-antenna users in their own cells and avoid the interference to the users served by other BSs in the same cluster.  In \cite{RobertNetworkMIMO:09,Jun2007:networkmimo}, the authors propose a BS cooperation strategy for fixed clusters, including full intra-cluster cooperation to eliminate intra-cluster interference and limited inter-cluster coordination to reduce the interference for the cluster edge users based on the per-cluster CSI and the CSI of the edge users in the neighboring clusters. In \cite{Howard2007:networkmimo,Brueck2010:networkmimo}, the authors consider different types of static cluster-based cooperation schemes  in a multi-cell system with multiple sectors per cell.}

However, all these works focus on physical layer performance (such as sum throughput, transport capacity) in cellular networks. 
They \textcolor{black}{ignore} the bursty  data arrivals and \textcolor{black}{assume} infinite
backlogs of packets at the transmitter. In other words, the
information flows are assumed to be delay insensitive. The resulting control policy is adaptive to the CSI only  and it cannot guarantee good delay performance for delay-sensitive applications  \cite{KBCellularCoodinateMIMO:2011,RobertNetworkMIMO:09,Jun2007:networkmimo,Howard2007:networkmimo,Brueck2010:networkmimo}. 
In practice,
a lot of applications have bursty arrivals and they are
delay-sensitive. It is very important to take into account the
delay performance in designing the cross-layer interference control
algorithms for the coordinated MIMO systems. 
The control policy for delay-sensitive applications should be adaptive
to both the CSI and the queue state information\footnote{The CSI
gives the knowledge about good opportunity to transmit whereas the
QSI gives the knowledge about the urgency of the data flow.} (QSI). 
\textcolor{black}{The motivation can be illustrated by the following example, as illustrated in Fig. \ref{fig:system-model} (a). Under cluster-based cooperative or coordinated MIMO, MSs only suffer from inter-cluster interference, as intra-cluster interference is eliminated. Therefore, cluster edge MSs suffer much more interference than cluster center MSs. In this work, we are interested to investigating delay-aware BS-discontinuous transmission (BS-DTX) control and user scheduling to reduce inter-cluster interference and save energy of the whole network. To maximize the sum throughput, the CSI-based BS-DTX control and user scheduling always favors cluster center MSs while starves cluster edge MSs. This may lead to infinite delay of cluster edge MSs and hence, infinite average delay of all the MSs. However, the QSI and CSI based design will dynamically favor different types of MSs to capture the urgency of data flows and the good opportunity of channels. Therefore, it can guarantee good delay performance.}
However, the design framework taking \textcolor{black}{into account} the queueing
delay and the physical layer performance is far from trivial as it
involves both queuing theory (to model the queuing dynamics) and information theory (to model the physical layer dynamics). 

In addition, recent initiatives towards green communications have
driven the design of wireless infrastructure to be more
energy-efficient. One energy-efficient design is to exploit
renewable energy at BSs. There are many recent works on power management  in
energy harvesting networks. For example, in
\cite{NeelyEHMobihoc11,TassiulasEH10}, the authors extend the
Lyapunov optimization framework  to derive an
efficient energy management algorithm for energy harvesting
networks. In  \cite{KarEH06}, the authors consider dynamic node
activation in energy harvesting sensor networks and propose a \textcolor{black}{simple} 
threshold-based node activation policy to achieve near-optimal system throughput.  Similarly, all these papers have focused on physical layer
throughput performance and the nodes are powered by renewable energy
source only with infinite energy storage size.

In this paper, we consider delay-optimal 
BS-DTX control and user scheduling
algorithm in downlink \textcolor{black}{energy harvesting} coordinated MIMO systems with limited
renewable energy storage. Each BS is powered by both conventional grid
and renewable power sources.  There are various first-order technical
challenges involved in solving the problem.

$\bullet$ {\bf Renewable and \textcolor{black}{Grid} Power Control:} The transmit power
of a BS comes from both renewable and \textcolor{black}{grid power} sources, which
have very different properties. For instance, the \textcolor{black}{grid power}  has
stable power supply but there is cost associated with it. On the
other hand, the renewable power is virtually free but it has random
supply and hence, an energy storage is needed for efficient
utilization of the renewable energy. In practice, the energy storage
has limited capacity and hence, the BS power control and user
scheduling algorithm should be adaptive to \textcolor{black}{the
renewable energy state information (ESI)} and the data QSI as well as the CSI. It is
highly nontrivial to strike a balance between these factors in the
control algorithm design.

$\bullet$ {\bf Delay-aware Low Complexity Distributed Algorithm:} While the \textcolor{black}{delay-optimal} control problem can be casted into an Markov Decision Process (MDP),
brute force solutions such as value iteration and policy iteration
 will suffer from the {\em curse of dimensionality}\cite{Bertsekas:2007}.
For example, a very large state space (exponential to the number of
users in the network) will be involved. In addition to the
complexity issue, the solution obtained will be centralized and it
requires knowledge of global system state information (ESI, QSI, CSI). However,
these system state information is usually distributed locally at
various BSs and huge signaling overhead will be involved in
collecting these information. Therefore, it is highly desirable to obtain a delay-aware low complexity and distributed algorithm with guaranteed delay performance.

$\bullet$ {\bf Performance Analysis:} Besides algorithm development, it is important to  analyze the system performance to understand how it is affected by the renewable energy storage size
and the interference coupling in cellular networks. One
challenge on the system performance analysis is  the statistical
characterization of  interference. In \cite{JeffSGCellularNet10},
the authors study the coverage and  rate of cellular networks without BS coordination using
stochastic geometry\cite{SGbookHaenggi08Now}. The locations of the
BSs are modeled as \textcolor{black}{a} homogeneous Poisson point process
(PPP) and the locations of the mobile stations (MSs) \textcolor{black}{are modeled} as some
independent (of the point process of BSs) point process. The analysis  
for coordinated MIMO network is more challenging due to the asymmetric topology induced by clustering. In addition, the analysis becomes more involved when queueing
dynamics of   data queues and renewable energy queues are considered.

In this paper, \textcolor{black}{considering the limited backhaul capacity and the latency in information exchange through backhaul in practical cellular systems \cite{Brueck2010:networkmimo}, we adopt cluster-based coordinated\footnote{\textcolor{black}{The design framework proposed in this paper does not rely on specific physical layer transmission schemes and can be easily extended to cluster-based cooperative MIMO.}} MIMO to eliminate intra-cluster interference.} We propose a {\em two-timescale} delay-aware 
BS-DTX control and user scheduling for energy harvesting downlink coordinated
MIMO systems as illustrated in Fig. \ref{fig:system-model} (a). The
BS-DTX control is  adaptive to both the ESI and the QSI over a
longer timescale. The user scheduling is adaptive to the ESI,
the QSI and the CSI over a shorter timescale. We show that the
two-timescale delay-optimal control problem can be modeled as an
infinite horizon average cost Partially Observed Markov Decision
Process (POMDP), which is well-known to be a difficult problem
\cite{Meuleau:1999}. By \textcolor{black}{using sample-path analysis} and exploiting the specific problem
structure, we first obtain some structural
results on the optimal control policy and derive an {\em equivalent
Bellman equation} with reduced state space.  
To derive a distributed control policy, we approximate the Q-factor and potential function associated with the equivalent Bellman equation by the {\em per-flow functions}. The per-flow functions are estimated online using distributed {\em stochastic learning} at each BS. 
We prove the almost-sure convergence of the
proposed distributed  algorithm. Furthermore, using
queueing theory, stochastic geometry and optimization techniques, we
characterize the sufficient conditions for data queues in the \textcolor{black}{coordinated} MIMO networks to be stable. 
Based on the analysis, we discuss the impacts of  the interference coupling and the size of renewable energy storage on network performance. Finally, we compare
the proposed algorithm with various baseline schemes and show that
significant delay performance gain can be achieved.

\section{System Models}\label{sec:system_model}
In this section, we shall elaborate \textcolor{black}{on} the 
system architecture, the physical layer model as well as the bursty source \textcolor{black}{model} for the coordinated MIMO networks.


\subsection{Architecture of \textcolor{black}{Downlink} Distributed MIMO Systems}
We consider a downlink coordinated MIMO system consisting of $B$ multi-antenna BSs and $K$ single-antenna MSs as illustrated in Fig. \ref{fig:system-model} (a). 
Each BS has $N_t$ transmit antennas. Let $\mathcal{K}_b$ denote the set of $K_b$ MS indices associated with the $b$-th BS and $\mathcal K$ denote the set of $K=\sum_b K_b$ MS indices in the network.  The set of BSs $\mathcal B = \{1,\cdots, B\}$ are partitioned into $N = B/N_t$ {\em coordination clusters}\footnote{For simplicity, we assume $B$ is a multiple of $N_t$.}, i.e., $\mathcal{B} = \cup_{n=1}^N \mathcal B_n$ and $\mathcal B_n\cap \mathcal B_{n'}=\emptyset $ $\forall n\neq n'$, where  $\mathcal{B}_n$ \textcolor{black}{denotes} the set of $B_n$ BSs in cluster $n$.  Each coordination cluster contains $N_t$ neighboring BSs and is managed by a cluster manager (CM) and all the $N$ CMs are managed by a BS controller (BSC). 
The BSs in the same cluster share the CSI and perform coordinated beamforming\cite{GesbertCoMIMOnewlook:2010} to combat intra-cluster interference. Besides
conventional \textcolor{black}{grid power} source, each BS is able to harvest energy from
the environment, \textcolor{black}{e.g.,} using solar panels
\cite{Gozalvezgreenradio10}. At each BS, there is a renewable energy
queue (battery) with limited capacity  for storing the harvested energy. In addition, at each BS, there are multiple data queues for buffering the packets to all
the MSs associated with the BS (one queue for each MS) as
illustrated in Fig. \ref{fig:system-model} (a).

\subsection{Physical Layer Model}

Let $\mathbf h_{k,b}\in \mathcal H $ and $L_{k,b}$ denote the $N_t\times 1$ complex 
small-scale fading vector  and the long-term path gain  between
the $b$-th  BS and the $k$-th  MS, where
$\mathcal H\subset \mathbb C^{N_t\times 1}$ denotes the finite discrete complex CSI state space.
Let $\mathbf H_{n}=\{\mathbf h_{k,b}: k\in \mathcal K_b, b\in \mathcal
B_n\}\in \boldsymbol{\mathcal H}_n \triangleq \mathcal H^{\sum_{b\in \mathcal B_n}K_b} $
and $\mathbf H=\cup_{n=1}^N \mathbf H_n\in \boldsymbol{\mathcal H}
\triangleq \mathcal H^{K}$ denote the
intra-cluster CSI  at $n$-th CM and the aggregation of \textcolor{black}{the} CSI
 over $N$ clusters, respectively. In this paper, the time
dimension is partitioned into scheduling slots indexed by $t$ with
slot duration $\tau$ (second).
\begin{Asump}[Quasi-static Fading]
$\mathbf h_{k,b}(t)$  is quasi-static in each scheduling slot for all
$(k,b)\in \mathcal K\times \mathcal B$. Furthermore, each element of vector $\mathbf h_{k,b}(t)$
follows a general distribution with mean 0 and vairiance 1. The distribution of each element of vector $\mathbf h_{k,b}(t)$  is i.i.d.  over scheduling \textcolor{black}{slots} and
independent w.r.t. $\{k,b\}$.  The
long-term path gain $L_{k,b}$ remains
constant for the duration of the communication session. 
~\hfill\QED\label{Asump:H}
\end{Asump}

We assume all the BSs in the system share a common spectrum.  Let  $p_b\in\mathcal P
\triangleq\{0,1\}$ denote the binary  BS-DTX control action of the $b$-th BS, where $p_b =1$ indicates the $b$-th BS is active and $p_b=0$ otherwise. Between the coordination clusters, the {\em
inter-cluster interference}  is managed by a binary {\em BS-DTX
control action}  $\mathbf p=\{p_b:p_b\in\mathcal
P,b\in \mathcal B \}\in\boldsymbol{\mathcal P}$, where
$\boldsymbol{\mathcal P} \subseteq \mathcal P^B$ is the aggregate BS-DTX
control action space and specifies the BS-DTX
patterns\cite{BPCmultiCELL}.
Since each BS has renewable and \textcolor{black}{grid power} sources, we have $
p_b=p_b^{E}+p_b^{G}$, where $p_b^{E}\in \mathcal P$ and $p_b^{G}
\in \mathcal P$ denote the power contribution from the renewable
power and \textcolor{black}{grid power} sources of the $b$-th BS, respectively. 
Let  $s_k\in\mathcal S\triangleq\{0,1\}$ denote the user scheduling action of the $k$-th MS, where $s_k=1$ indicates the $k$-th MS is selected  to receive packets and $s_k=0$ otherwise. Thus, users  are selected according to a
{\em user scheduling action}  $\mathbf
s=\{s_k:s_k\in\mathcal S,k\in \mathcal K \}\in \boldsymbol{\mathcal
S}$, where  $ \boldsymbol{\mathcal
S} \subseteq \mathcal S^K$ is the aggregate user scheduling action space.
The BS-DTX control and user scheduling are performed according to a
control policy to be defined in Definition
\ref{Defn:feasible_policy}.

In each slot, each active BS selects one MS to serve. Within each coordination cluster, the active BSs combat the intra-cluster interference using coordinated beamforming \cite{GesbertCoMIMOnewlook:2010,WeiYuCoordinatedBF:2010,KBCellularCoodinateMIMO:2011}. Let
$P_b$ and $x_k$ denote the instantaneous transmit power of
the $b$-th BS and the information symbols for the $k$-th MS,
respectively. The received signal \textcolor{black}{at} the $k$-th MS of the $b$-th
cell in the $n$-th cluster is given by
\begin{align}
y_k = &\underbrace{ p_b \sqrt{P_b} \sqrt{L_{k,b}} \mathbf h_{k,b}^T
\mathbf w_{k,b}  s_k x_k}_{\text{desired signal}} +
\underbrace{  \sum_{\substack{b'\in \mathcal B_n, b'\neq b}} p_{b'} \sqrt{P_{b'}}
\sqrt{L_{k,b'}} \mathbf h_{k,b'}^T \left(\sum_{\substack{k'\in \mathcal K_{b'}}}\mathbf w_{k',b'}
s_{k'}x_{k'}\right)}_{\text{intra-cluster interference}}
\nonumber\\
 &+ \underbrace{\sum_{\substack{n'\neq n}}\sum_{\substack{b'\in \mathcal B_{n'}}} p_{b'} \sqrt{P_{b'}}
\sqrt{L_{k,b'}} \mathbf h_{k,b'}^T \left(\sum_{\substack{k'\in \mathcal K_{b'}}}\mathbf w_{k',b'}
s_{k'}x_{k'}\right)}_{\text{inter-cluster interference}} +
\underbrace{z_k}_{\text{noise}}, \  k \in \mathcal K_b, b\in \mathcal B_n
\nonumber
\end{align}
where $z_k \sim \mathcal{CN}(0,N_0)$ is the AWGN noise and $\mathbf w_{k,b}\in\mathbb C^{N_t\times 1}$
is the {\em zero-forcing beamforming weight}  for  the $k$-th MS at the $b$-th BS.
Specifically, 
$\{\mathbf w_{k,b}\}$ is given by the solution\footnote{If there are more than one solutions, we choose the one maximizes $||\mathbf h_{k,b}^T \mathbf  w_{k,b} ||^2$.} of the zero-forcing problem: 
$\sum_{k \in
\mathcal K_b} ||\mathbf w_{k,b}||^2 s_k = p_b$  and $s_k\mathbf h_{k,b'}^T \left(\sum_{\substack{k'\in \mathcal K_{b'}}}\mathbf w_{k',b'}
s_{k'}\right)=0$ ($\forall b'\in \mathcal B_n, b'\neq b$).

The receive SINR at the $k$-th MS of the $b$-th
cell in the $n$-th cluster is given by
\begin{align}
\rho_k(\mathbf H, \mathbf p, \mathbf s) = \frac{P^{rx}_k} {N_0+I_k},\ k \in \mathcal K_b, b\in \mathcal B_n
\label{eqn:SINR_per_user}
\end{align}
where the receive power $P^{rx}_k$ and the inter-cluster
interference power $I_k$ are given by
\begin{align}P^{rx}_k&= p_b P_b L_{k,b}||\mathbf h_{k,b}^T \mathbf  w_{k,b} ||^2
s_k\label{eqn:pwr-k}
\end{align}
\begin{align}
I_k&=\sum_{\substack{n'\neq n}}\sum_{\substack{b'\in \mathcal B_{n'}}} p_{b'} P_{b'}
L_{k,b'} ||\mathbf h_{k,b'}^T \big(\sum_{\substack{k'\in \mathcal K_{b'}}}\mathbf w_{k',b'}
s_{k'}\big)||^2\label{eqn:inf-k}
\end{align}
We have the following assumption regarding packet transmission.
\begin{Asump} [Packet Transmission Model] One data packet with
certain fixed packet size can be successfully received by the $k$-th
MS  if the receive SINR $\rho_k$ exceeds a certain
threshold\footnote{In general, we allow different MSs with different
packet sizes, and hence the threshold is indexed by $k$ and may be
different for different MSs.} $\delta_k$, i.e., $\rho_k\geq
\delta_k$. There exists a state-action pair $(\mathbf H, \mathbf p, \mathbf s)\in \boldsymbol{\mathcal H}\times \boldsymbol{\mathcal P}\times \boldsymbol{\mathcal S}$, such that $\Pr[\rho_k(\mathbf H, \mathbf p, \mathbf s)\geq \delta_k]>0$.
~\hfill\QED \label{Asump:pck-tx-model}
\end{Asump}

\subsection{Bursty Source Model and Queue Dynamics}

Let $\mathbf{A}^Q(t)=\{A^Q_k(t):  k\in \mathcal{K}\}$ and
$\mathbf{A}^E(t)=\{A^E_b(t):  b\in \mathcal{B}\}$ be the number of
packets \textcolor{black}{arriving to} the $K$ MSs and the number of renewable energy
units\footnote{One unit of energy for the $b$-th BS corresponds to
the amount of energy consumed in downlink transmission at each slot
for the $b$-th BS, i.e. $P_b\tau$ Joule. \textcolor{black}{
Note that the instantaneous transmit power from the renewable power source is finite (i.e., $P_b$). 
The notion ``unit of energy'' can be easily extended from binary (on-off) power control to handle (multi-level) power control. 
}
} \textcolor{black}{arriving to} the $B$ BSs at the end of the $t$-th scheduling slot, respectively. We have the
following assumptions\footnote{Note that under Assumption \ref{Asump:general_A-Q} and Assumption \ref{Asump:general_A-E}, we have $\Pr[A^Q_k(t)=0]>0$ and $\Pr[A^E_b(t)=0]>0$ for all $k\in \mathcal K$ and $b\in \mathcal B$, respectively.} regarding the bursty data and renewable
energy arrival processes.
\begin{Asump} [Bursty Data Source Model]
The arrival process  $A^Q_k(t)$ is  i.i.d. over scheduling \textcolor{black}{slots}
and independent w.r.t.  $k$ according to a general distribution
$P_{A^Q_k}(\cdot)$  with average arrival rate $\mathbb E[A^Q_k(t)] =
\lambda^Q_k<1$. The statistics of $A^Q_k(t)$ is   unknown to the
controller.~\hfill\QED \label{Asump:general_A-Q}
\end{Asump}
\begin{Asump} [Bursty Renewable Energy Model] The arrival process  $A^E_b(t)$ is  i.i.d. over scheduling \textcolor{black}{slots}
and independent w.r.t.  $b$ according to a general distribution
$P_{A^E_b}(\cdot)$  with average arrival rate $\mathbb E[A^E_b(t)] =
\lambda^E_b<1$. The statistics of $A^E_b(t)$ is  unknown to the
controller.~\hfill\QED\label{Asump:general_A-E}
\end{Asump}

\textcolor{black}{
\begin{Rem} [Interpretation of Assumption \ref{Asump:general_A-E}]
Assumption \ref{Asump:general_A-E} implies that the renewable power source is stationary. Although the renewable energy source is not stationary over a very long time horizon in practice,  it  is stationary over a typical communication session, which lasts for \textcolor{black}{less} than 30 mins.~\hfill\QED
\end{Rem}
}

Let $\mathbf{Q}_n(t)=\{Q_k(t):k \in \mathcal K_n\}\in
\boldsymbol{\mathcal Q}_n\triangleq \mathcal Q^{\sum_{b\in \mathcal B_n}K_b}$ be the $n$-th
cluster QSI  and $\mathbf{Q}(t)=\cup_{n=1}^N \mathbf{Q}_n(t)
\in \boldsymbol{\mathcal Q}\triangleq \mathcal Q^{K}$ be the
aggregation of \textcolor{black}{the} QSI over $N$ clusters at the beginning
of the $t$-th slot, where $Q_k(t)\in \mathcal Q\triangleq \{0, 1,
\cdots, N_Q\}$ denotes the number of data packets at the data queue
for the $k$-th MS and $N_Q$ denotes the data buffer size. At slot
$t$, there is $\mathbf I[\rho_k(t)\geq \delta_k]\in \{0,1\}$ packet
successfully received at the $k$-th MS, where $\mathbf I[\cdot]$
denotes the indicator function. Hence, the data queue dynamics of
the $k$-th MS is given by
\begin{align}
Q_k(t + 1) = \min \Big \{\big[Q_k(t)-\mathbf I[\rho_k(t)\geq
\delta_k]\big]^+ + A^Q_k(t),\ N_Q \Big\},\ k \in \mathcal K
\label{eqn:data-queue-dyn}
\end{align}
where $\rho_k(t)\triangleq\left(\mathbf H(t), \mathbf p(t), \mathbf s(t)\right)$ and $x^+\triangleq\max\{x,0\}$.

Similarly, let $\mathbf{E}_n(t)=\{E_b(t):b \in \mathcal B_n\}\in
\boldsymbol{\mathcal E}_n\triangleq \mathcal E^{B_n}$ be the $n$-th
cluster ESI  and $\mathbf{E}(t)=\cup_{n=1}^N \mathbf{E}_n(t)
\in \boldsymbol{\mathcal E}\triangleq \mathcal E^B$ be the
aggregation of the \textcolor{black}{ESI} over $N$ clusters at the beginning
of the $t$-th slot, where $E_b(t)\in \mathcal E\triangleq \{0, 1,
\cdots, N_E\}$ denotes the number of renewable energy units in the
energy queue for the $b$-th BS and $N_E$ denotes the energy storage
size. At slot $t$, $p_b^E(t)\in \mathcal P$ unit of renewable energy
is consumed from the $b$-th energy queue for packet transmission.
Hence, the energy queue dynamics of the $b$-th BS is given by
\begin{align}
E_b(t + 1) = \min \Big \{\big[E_b(t)-p_b^E(t)\big]^+ + A^E_k(t),\
N_E \Big\},\ b \in \mathcal B \label{eqn:energy-q-dyn}
\end{align}

\subsection{BS-DTX Control and User Scheduling Policy}

For notation convenience, we denote
$\boldsymbol{\chi}(t)=\big(\mathbf{E}(t), \mathbf{Q}(t),
\mathbf{H}(t)\big)\in \boldsymbol{\mathcal X}=  \boldsymbol{\mathcal
E}\times\boldsymbol{\mathcal Q }\times \boldsymbol{\mathcal H}$ as
the {\em global system state} at the $t$-th slot. We first define
the centralized control policy. Specifically, at the beginning of each slot, the
controller determines the {\em renewable power DTX control action}
$\mathbf p^E=\{p_b^E:p_b^E\in\mathcal P,b\in \mathcal B
\}\in\boldsymbol{\mathcal P}$, {\em \textcolor{black}{grid power} DTX control action}
$\mathbf p^{G}=\{p_b^{G}:p_b^{G}\in\mathcal P,b\in \mathcal B
\}\in\boldsymbol{\mathcal P}$  as well as the {\em user scheduling
action}  $\mathbf s=\{s_k:s_k\in\mathcal S,k\in \mathcal K \}\in
\boldsymbol{\mathcal S}$  based on the global system state
$\boldsymbol{\chi}(t)$ according to the control \textcolor{black}{policy} defined
below.

\begin{Def}[BS-DTX Control and User Scheduling Policy] A BS-DTX control and
user scheduling policy consists of a sequence of mappings
$\pi=\{\Omega^1,\Omega^2,\cdots\}$.  The mapping for the $t$-th slot
$\Omega^t=(\Omega_p^{E,t},\Omega_p^{G,t},\Omega_s^t)$ 
is a mapping from the system state $\boldsymbol{\chi}(t)\in
\boldsymbol{\mathcal X}$ to the renewable power DTX control action
$\Omega_p^{E,t}(\mathbf E(t),\mathbf{Q}(t))= \mathbf p^E(t)\in
\boldsymbol{\mathcal P}$, the \textcolor{black}{grid power} DTX control action
$\Omega_p^{G,t}(\mathbf E(t),\mathbf{Q}(t))= \mathbf p^{G}(t)\in
\boldsymbol{\mathcal P}$ and the user scheduling action
$\Omega_s^t(\boldsymbol{\chi}(t))=\mathbf s(t)\in
\boldsymbol{\mathcal S}$. A policy $\pi$ is called {\em feasible} if
for all $t$, the following constraints are satisfied:
\begin{enumerate}
\item $p_b^E(t)=0$ if $E_b(t)=0$ for all $b\in \mathcal B$ (no renewable energy available for transmission).
\item $p_b(t)=p_b^E(t)+p_b^{G}(t)\in \mathcal P$ for all  $b \in \mathcal B$ (binary BS-DTX  control).
\item $\sum_{k \in \mathcal K_b}s_k(t) =p_b(t)$ for all $b\in \mathcal B$ (each active BS selects one MS in its cell).~\hfill\QED
\end{enumerate}
\label{Defn:feasible_policy}
\end{Def}

\textcolor{black}{
\begin{Rem} [Motivation of Two-Timescale Control Policy]
The two-timescale control is a constraint we impose due to the following practical reasons. 
The QSI and ESI are changing on a longer timescale (e.g., several slots) while the CSI is changing on a shorter timescale (e.g., per-slot). The BS-DTX control is usually implemented at the BSC for interference reduction and energy saving of the whole network. As a result, the BS-DTX control cannot afford to be running on a per-slot basis, due to the high complexity and signaling overhead in collecting the local CSI from all the BSs. Therefore, it is desirable to make it a function of the ESI and QSI only.
On the other hand, the low complexity distributed user scheduling is implemented locally at each CM (similar to HSDPA in current 3G networks) and they can afford to run on a per-slot basis and adapt to the ESI, QSI and CSI.~\hfill\QED
\end{Rem}
}

\section{Problem Formulation and Optimal Solution}\label{sec:CMDP_formulation}
In this section, we shall first elaborate \textcolor{black}{on} the dynamics of the system
state under a control policy $\pi$. Based on that, we shall 
formulate the delay-optimal control problem and derive some
structural properties for the optimal solution.

\subsection{Delay-Optimal  Problem Formulation}

Under Assumptions \ref{Asump:H}, 
\ref{Asump:general_A-Q} and \ref{Asump:general_A-E}, the induced random process
$\{\boldsymbol{\chi}(t)\}$ for a given feasible control policy
$\pi=\{\Omega^1,\Omega^2,\cdots \}$ is a Markov chain with the
following transition probability
\begin{align}
&\Pr[\boldsymbol{\chi}(t+1)|\boldsymbol{\chi}(t),\Omega^t(\boldsymbol{\chi}(t))]\nonumber\\
=&
\Pr[\mathbf{H}(t+1)|\boldsymbol{\chi}(t),\Omega^t(\boldsymbol{\chi}(t))]\Pr[\mathbf{E}(t+1)|\boldsymbol{\chi}(t),\Omega^t(\boldsymbol{\chi}(t))]
\Pr[\mathbf{Q}(t+1)|\boldsymbol{\chi}(t),\Omega^t(\boldsymbol{\chi}(t))]\nonumber\\
=&\Pr[\mathbf{H}(t+1)]
\Pr[\mathbf{E}(t+1)|\boldsymbol{\chi}(t),\Omega^t(\boldsymbol{\chi}(t))]
\Pr[\mathbf{Q}(t+1)|\boldsymbol{\chi}(t),\Omega^t(\boldsymbol{\chi}(t))]\label{eqn:transition-prob1}
\end{align}

As a result, given a feasible control policy $\pi$, the average
delay cost per stage of the $k$-th MS starting from a given initial
state $\boldsymbol \chi(1)$ is given by
\begin{equation}
\overline{D_{\pi,k}}\big(\boldsymbol \chi(1)\big) =
\limsup_{T\rightarrow\infty}
\frac{1}{T}\mathbb{E}^{\pi}\left[\sum_{t=1}^T
f\big(Q_k(t)\big)\right],\ \forall k \in \mathcal{K}
\label{eqn:delay1}
\end{equation}
where the expectation is taken w.r.t. the measure induced by the
policy $\pi$ and $f(Q_k)$ is a monotonic increasing utility function
of $Q_k$. For example,  with $f(Q_k)=\frac{Q_k}{\lambda_k}$ and
$f(Q_k)=\mathbf{1}[Q_k\geq Q_k^o]$ ($Q_k^o \in \{0,\cdots,N_Q\}$),
\eqref{eqn:delay1} can be used to measure the average delay and the
average queue outage probability of the $k$-th MS under policy
$\pi$. Similarly, given a feasible control policy $\pi$, the average
\textcolor{black}{grid power} cost per stage of the $b$-th BS starting from a given
initial state $\boldsymbol \chi(1)$  is given by
\begin{equation}
\overline{p_{\pi,b}^{G}}\big(\boldsymbol \chi(1)\big)
=\limsup_{T\rightarrow\infty}
\frac{1}{T}\mathbb{E}^{\pi}\left[\sum_{t=1}^T p_b^{G}(t)\right],\
\forall b \in \mathcal{B} \label{eqn:pwr-cost}
\end{equation}
We are interested in minimizing the average delay cost of each MS
$k\in \mathcal K$ in \eqref{eqn:delay1} and the average \textcolor{black}{grid power}
cost of each BS $b\in \mathcal B$ in \eqref{eqn:pwr-cost}. 
A Pareto optimal tradeoff on the average delay and average \textcolor{black}{grid power}
consumption can be obtained by solving the following problem. 
\begin{Prob}[\textcolor{black}{Two-Timescale Delay-Optimal Control}]
For some positive constants $\boldsymbol{\beta}=\{\beta_k>0: k\in
\mathcal{K}\}$ and $\boldsymbol{\gamma}=\{\gamma_b>0: b\in
\mathcal{B}\}$, the delay-optimal problem is formulated as
\begin{eqnarray}
\min_{\pi}J_{\pi}^{(\boldsymbol{\beta},\boldsymbol{\gamma})}\big(\boldsymbol
\chi(1)\big)=\sum_{k\in \mathcal
K}\beta_k\overline{D_{\pi,k}}\big(\boldsymbol
\chi(1)\big)+\sum_{b\in \mathcal
B}\gamma_b\overline{p_{\pi,b}^{G}}\big(\boldsymbol \chi(1)\big)=
\lim_{T\rightarrow\infty} \frac{1}{T}
\mathbb{E}^{\pi}\left[\sum_{t=1}^T
g\big(\boldsymbol{\chi}(t),\Omega^t(\boldsymbol{\chi}(t))\big)\right]
\label{cons-MDP}
\end{eqnarray}
where
$g\big(\boldsymbol{\chi}(t),\Omega^t(\boldsymbol{\chi}(t))\big)=\sum_{k\in
\mathcal K} \beta_k f\big(Q_k(t)\big)+\sum_{b\in \mathcal B}\gamma_b
p_b^{G}(t)$  \textcolor{black}{and the control policy $\pi$ satisfies the two-timescale requirement in Definition \ref{Defn:feasible_policy}.} ~\hfill\QED\label{Prob2}
\end{Prob}

\textcolor{black}{
\begin{Rem}[Two-Timescale Control and POMDP]
By two-timescale requirement, the BS-DTX  control policy is defined on the
partial system state $(\mathbf E,\mathbf Q)$, while the user
scheduling policy is defined on the complete system state
$\boldsymbol \chi =(\mathbf E, \mathbf Q, \mathbf H)$.  Due to the two-timescale control constraint as in Definition \ref{Defn:feasible_policy}, Problem \ref{Prob2} is a POMDP\footnote{POMDP is an extension of MDP when the control agent does not have direct observation of the entire system state.}.~\hfill\QED
\end{Rem}
}

\subsection{\textcolor{black}{Policy and State Space Reduction}}

\textcolor{black}{Problem \ref{Prob2} belongs to POMDP, which is well-known to be a challenging problem in general. Yet, we shall exploit some special structures in our problems to reduce the policy and state spaces. Based on that, we can simplify the POMDP. We first have the following lemma on the structural property of the BS-DTX control, which helps to reduce the policy space.}


\begin{Lem} [Structure of Optimal BS-DTX Control]
Let the BS-DTX control  for the $t$-th slot be denoted by
$\Omega_p^{t}:\boldsymbol{\mathcal E}\times\boldsymbol{\mathcal
Q}\to \boldsymbol{\mathcal P}$, which is a mapping from the partial system
state $(\mathbf E,\mathbf Q)\in \boldsymbol{\mathcal E}\times\boldsymbol{\mathcal
Q}$ to the BS-DTX
control action $\Omega_p^{t}(\mathbf E(t),\mathbf{Q}(t))= \mathbf
p(t)\in \boldsymbol{\mathcal P}$. Conditioned on any $\Omega_{p,b}^{t}$,
the optimal $\Omega_{p,b}^{E,t}$ and $\Omega_{p,b}^{G,t}$ satisfy
$\Omega_{p,b}^{E,t}(\mathbf
E(t),\mathbf{Q}(t))=\Omega_{p,b}^{t}(\mathbf
E(t),\mathbf{Q}(t))\mathbf I[E_b(t)>0]$ and 
$\Omega_{p,b}^{G,t}(\mathbf
E(t),\mathbf{Q}(t))=\Omega_{p,b}^{t}(\mathbf
E(t),\mathbf{Q}(t))\mathbf I[E_b(t)=0]$ for all $b\in \mathcal B$
and all $t$.~\hfill\QED\label{Lem:opt-re-ac-pwr-control}
\end{Lem}
\begin{proof} Please refer to Appendix A.
\end{proof}

\begin{Rem}[Interpretation of Lemma \ref{Lem:opt-re-ac-pwr-control}]
Lemma \ref{Lem:opt-re-ac-pwr-control}
indicates that we are inclined to consume renewable power first.
This is because renewable power is free while \textcolor{black}{grid power} has cost. In
addition, due to the limited energy storage size, we may suffer
from renewable energy loss when the energy queue size is large.
Therefore, it is preferable to keep the size of the energy queue
small.~\hfill\QED
\end{Rem}

Based on Lemma \ref{Lem:opt-re-ac-pwr-control}, without loss of
optimality, we can first solve Problem \ref{Prob2} over a reduced
policy $\pi=\{\Omega^1,\Omega^2,\cdots\}$, where
$\Omega^t=(\Omega_p^{t},\Omega_s^{t})$ and then obtain the optimal
$\Omega_p^{E,t}$ and $\Omega_p^{G,t}$ from the optimal
$\Omega_p^{t}$ using Lemma \ref{Lem:opt-re-ac-pwr-control}.

\textcolor{black}{Next, we exploit the i.i.d. property of the CSI to reduce the state space.} We first define {\em partitioned
actions} below:

\begin{Def}[Partitioned Actions]
Given  $\Omega^t=(\Omega_p^t,\Omega_s^t)$, we define
$$\Omega^t(\mathbf E,\mathbf Q)=\{(\mathbf p,\mathbf s) =\big(
\Omega_p^t(\mathbf E,\mathbf Q),\Omega_s^t(\mathbf E,\mathbf Q,
\mathbf H)\big): \mathbf{H}\in \boldsymbol{\mathcal H}\}, \
\Omega_s^t(\mathbf E,\mathbf Q)=\{\mathbf s =\Omega_s^t(\mathbf
E,\mathbf Q, \mathbf H): \mathbf{H}\in \boldsymbol{\mathcal H}\}$$
as the collection of actions $(\mathbf p,\mathbf s)$ and $\mathbf s$
for all possible CSI $\mathbf{H}$ conditioned on a given ESI and QSI
pair $(\mathbf E,\mathbf Q)$.  $\Omega^t$ and $\Omega_s^t$ are
therefore equal to the union of all partitioned actions. i.e.
$\Omega= \bigcup_{(\mathbf E,\mathbf Q)}\Omega(\mathbf E,\mathbf Q)$
and $\Omega_s= \bigcup_{(\mathbf E,\mathbf Q)}\Omega_s(\mathbf
E,\mathbf Q)$. \label{defn:conditional-action1} ~ \hfill\QED
\end{Def}

Based on \textcolor{black}{Lemma \ref{Lem:opt-re-ac-pwr-control} and} Definition \ref{defn:conditional-action1}, the optimal
control policy in Problem \ref{Prob2} can be obtained by solving
an {\em equivalent Bellman equation} \textcolor{black}{over a reduced state space}, which is summarized in the
lemma below.

\begin{Lem}[Equivalent Bellman Equation for POMDP]
The optimal control policy for Problem \ref{Prob2} can be obtained
by solving the  following {\em equivalent Bellman equation} w.r.t.
$\big(\theta, \{V(\mathbf E,\mathbf Q)\}\big)$:
\begin{align}
\theta+ V(\mathbf E, \mathbf Q)=\min_{\Omega(\mathbf E,\mathbf
Q)}\Big\{ g \big((\mathbf E,\mathbf Q), \Omega(\mathbf E,\mathbf
Q)\big) + \sum_{(\mathbf E',\mathbf Q')} \Pr[(\mathbf E',\mathbf
Q')| &(\mathbf E,\mathbf Q), \Omega(\mathbf E,\mathbf Q)]V(\mathbf
E',\mathbf Q')\Big\}\nonumber\\
 &\forall (\mathbf E,\mathbf Q)\in \boldsymbol{\mathcal E}\times\boldsymbol{\mathcal
 Q}
 \label{eqn:Bellman3}
\end{align}
where $g \big((\mathbf E,\mathbf Q), \Omega(\mathbf E,\mathbf
Q)\big)=\sum_{k\in \mathcal K}\beta_k f(Q_k)+\sum_{b\in \mathcal
B}\gamma_b \Omega_p(\mathbf E,\mathbf Q)\mathbf
I[E_b=0]$ is the per-stage cost function,
$ \Pr[(\mathbf E',\mathbf Q')| (\mathbf E,\mathbf Q), \Omega(\mathbf
E,\mathbf Q)]=\mathbb E \big[\Pr[(\mathbf E',\mathbf Q')|\boldsymbol
\chi, \Omega(\boldsymbol \chi) ]\big|(\mathbf E,\mathbf Q) \big]$ is
the transition kernel.  $\theta$ is the optimal value for
all $\boldsymbol \chi$, i.e.,
$\theta=\min_{\pi}J_{\pi}^{(\boldsymbol{\beta},\boldsymbol{\gamma})}\big(\boldsymbol
\chi\big)$ $\forall \boldsymbol \chi \in \boldsymbol {\mathcal X}$
and $\{V(\mathbf E,\mathbf Q)\}$ is called the {\em potential
function}. Furthermore, if $\Omega^*(\mathbf E,\mathbf
Q)=\big(\Omega_p^*(\mathbf E,\mathbf Q), \Omega_s^*(\mathbf
E,\mathbf Q)\big)$ attains the minimum of the R.H.S. of
\eqref{eqn:Bellman3} for all $ (\mathbf E,\mathbf Q)\in
\boldsymbol{\mathcal E}\times \boldsymbol{\mathcal Q}$, the
stationary policy $\Omega^*=(\Omega_p^*, \Omega_s^*)$ is optimal
(i.e., $\pi^*=\{\Omega^*,\Omega^*,\cdots\}$).~\hfill\QED
\label{Lem:reduced-MDP}
\end{Lem}
\begin{proof}
Please refer to the Appendix B.
\end{proof}

\begin{Rem} [Interpretation of Equivalent Bellman Equation] The equivalent Bellman equation in \eqref{eqn:Bellman3} is defined on  \textcolor{black}{the reduced space of the} ESI and QSI $(\mathbf E,\mathbf Q)$ only. Nevertheless, by solving \eqref{eqn:Bellman3}, we can 
obtain a stationary BS-DTX  policy $\Omega_p^*$\textcolor{black}{, which} is a function
of (ESI, QSI),  and a stationary user scheduling policy $\Omega_s^*$\textcolor{black}{, which}  is a function of (ESI, QSI, CSI).~\hfill\QED
\end{Rem}

\subsection{Centralized Optimal BS-DTX Control and User Scheduling}

To facilitate the BS-DTX control, which is only adaptive to the ESI
and the QSI, we introduce the {\em BS-DTX control Q-factor} $\mathbb
Q (\mathbf E,\mathbf Q, \mathbf p)$ w.r.t. the BS-DTX control action
$\mathbf p$. Based on Lemma \ref{Lem:reduced-MDP}, we summarize the
optimal BS-DTX control in the following corollary.
\begin{Cor} [Optimal BS-DTX Control]
The optimal BS-DTX control   is given by
\begin{align}
\Omega_p^*(\mathbf E, \mathbf Q)=\arg \min_{\mathbf p \in \mathcal
P} \mathbb Q (\mathbf E,\mathbf Q, \mathbf p), \ \forall (\mathbf
E,\mathbf Q)\in \boldsymbol{\mathcal E}\times \boldsymbol{\mathcal
Q} \label{eqn:opt-pattern-selection-Q-factor}
\end{align}
where $\mathbb Q (\mathbf E,\mathbf Q, \mathbf p)$ is the BS-DTX
control Q-factor given by  the  following Bellman equation w.r.t.
$\big(\theta, \{\mathbb Q(\mathbf E,\mathbf Q,\mathbf p)\}\big)$:
\begin{align}
&\theta+\mathbb Q (\mathbf E,\mathbf Q, \mathbf p)\hspace{70mm}
\forall (\mathbf E,\mathbf Q)\in \boldsymbol{\mathcal
E}\times\boldsymbol{\mathcal Q}, \mathbf p \in
\boldsymbol{\mathcal P}\label{eqn:Bellman-Q-factor}\\
=&\min_{\Omega_s(\mathbf E,\mathbf Q)}\Big\{ g \big((\mathbf
E,\mathbf Q), \mathbf p, \Omega_s(\mathbf E,\mathbf Q)\big) +
\sum_{(\mathbf E',\mathbf Q')} \Pr[(\mathbf E',\mathbf Q')| (\mathbf
E,\mathbf Q), \mathbf p, \Omega_s(\mathbf E,\mathbf Q)]\min_{\mathbf
p' \in \boldsymbol{\mathcal P}} \mathbb Q (\mathbf E',\mathbf Q',
\mathbf p')\Big\}\nonumber
\end{align}~\hfill\QED\label{Cor:Bellman-Q-factor}
\end{Cor}
\begin{proof}
Please refer to Appendix B.
\end{proof}

As the distributions of the energy 
and \textcolor{black}{data} arrival processes are unknown to the controllers,  we
introduce the post-decision state potential function $U(\widetilde{\mathbf E},\widetilde{\mathbf Q})$  to determine the
user selection \cite{Thesis:Salodkar}. The post-decision
state $(\widetilde{\mathbf E},\widetilde{\mathbf Q})$ is defined to
be the virtual partial system state  immediately after making an action
before the new renewable energy and data arrive\footnote{For
example, $\boldsymbol \chi=(\mathbf E,\mathbf{Q},\mathbf{H})$ is the
state at the beginning of some slot (also called the pre-decision
state) and making an action $(\mathbf p, \mathbf
s)=\Omega(\boldsymbol \chi)$ leads to $\boldsymbol \rho=\{\rho_k:
k\in \mathcal K\}$ with $\rho_k$ given by \eqref{eqn:SINR_per_user}.
Then, the post-decision state immediately after the action is
$\widetilde{\boldsymbol \chi}=(\widetilde {\mathbf E}, \widetilde
{\mathbf Q},\mathbf H)$, where $\widetilde {\mathbf E}=[\mathbf
E-\mathbf p]^+$ and $\widetilde {\mathbf Q}=\big[\mathbf Q-\mathbf
I[\boldsymbol \rho \succeq\boldsymbol \delta]\big]^+$, where
$\boldsymbol\delta=\{\delta_k:k\in \mathcal K\}$. If new arrivals
$\mathbf A^E$ and $\mathbf A^Q$ occur in the post-decision state,
and the CSI changes to $\mathbf H'$, then the system reaches the
next actual state, i.e., pre-decision state $\boldsymbol
\chi'=(\min\{\widetilde {\mathbf E}+\mathbf A^E, N_E\},
\min\{\widetilde {\mathbf Q}+\mathbf A^Q, N_Q\},\mathbf H')$.}.
Based on Lemma \ref{Lem:reduced-MDP}, we summarize the optimal user
scheduling in the following corollary.
\begin{Cor} [Optimal User Scheduling]
The optimal user scheduling  is given by
\begin{align}
\Omega_s^*(\boldsymbol \chi)=\arg \min_{\mathbf s\in
\boldsymbol{\mathcal S}(\mathbf p^*)}\Big\{\sum_{\mathbf d\in
\boldsymbol{\mathcal D}}\Big(\prod_{k\in \mathcal
K}\left(1-d_k-(-1)^{d_k}\Pr[\rho_k (\mathbf H, \mathbf p^*, \mathbf
s)\geq \delta_k]\right)&U\big([\mathbf E-\mathbf
p^*]^+,[\mathbf Q-\mathbf d]^+\big)\Big)\Big\}\nonumber\\
& \forall \boldsymbol{\chi}\in \boldsymbol{\mathcal X}
\label{eqn:opt-user-selection-postd-potential}
\end{align}
where  $\mathbf p^*=\Omega_p^*(\mathbf E,\mathbf Q)$ is the optimal
BS-DTX control action given by
\eqref{eqn:opt-pattern-selection-Q-factor}, $\boldsymbol{\mathcal
S}(\mathbf p)\triangleq \{\mathbf s\in \boldsymbol{\mathcal S}:
\sum_{k \in \mathcal K_b}s_k = p_b,
b\in\mathcal B\}$ denotes the feasible user scheduling action space under
the BS-DTX control action $\mathbf p$, $d_k\in \mathcal
D\triangleq\{0,1\}$, and $\mathbf d=\{d_k\in \mathcal D_k:k\in
\mathcal K\}\in \boldsymbol{\mathcal D}\triangleq \mathcal D^K$.
$U(\widetilde{\mathbf E},\widetilde{\mathbf Q})$ is  the
post-decision potential function given by the  following Bellman equation w.r.t.
$\big(\theta, \{U(\widetilde{\mathbf E},\widetilde{\mathbf Q})\}\big)$\cite{Thesis:Salodkar}:
\begin{align}
&\theta+ U(\widetilde{\mathbf E},\widetilde{\mathbf
Q})\hspace{80mm} \forall (\widetilde{\mathbf E},\widetilde{\mathbf
Q})
\in \boldsymbol{\mathcal E}\times\boldsymbol{\mathcal Q} \label{eqn:Bellman-postd}\\
=&\sum_{\mathbf A^E,\mathbf A^Q}\Pr[\mathbf A^E]\Pr[\mathbf
A^Q]\min_{\Omega(\mathbf E,\mathbf Q)}\Big\{ g \big((\mathbf
E,\mathbf Q), \Omega(\mathbf E,\mathbf Q)\big) +
\sum_{(\widetilde{\mathbf E}',\widetilde{\mathbf Q}')}
\Pr[(\widetilde{\mathbf E}',\widetilde{\mathbf Q}')| (\mathbf
E,\mathbf Q), \Omega(\mathbf E,\mathbf Q)]U(\widetilde{\mathbf E}',\widetilde{\mathbf Q}')\Big\}\nonumber
\end{align}
where $\mathbf E=\min\{\widetilde {\mathbf E}+\mathbf A^E, N_E\}$
and $\mathbf Q=\min\{\widetilde {\mathbf Q}+\mathbf A^Q,
N_Q\}$.~\hfill\QED\label{Cor:Bellman-postd}
\end{Cor}
\begin{proof}
Please refer to Appendix B.
\end{proof}

\begin{Rem} [Complexity of Centralized Delay-Optimal Solution]
The complexity of obtaining the original Q-factor and the associated
BS-DTX control is $\mathcal O \big((N_E+1)^B(N_Q+1)^K 2^B\big)$. The
complexity of obtaining the original post-decision state potential
function and the associated user scheduling is $\mathcal O
\big((N_E+1)^B(N_Q+1)^K\big)$.~\hfill\QED
\end{Rem}

\section{Low Complexity Delay-aware Distributed Solution}\label{sec:general_decentral_algo}

Obtaining  the optimal control in \eqref{eqn:opt-pattern-selection-Q-factor} and \eqref{eqn:opt-user-selection-postd-potential}  has exponential complexity and
requires centralized implementation at the BSC and knowledge of
the aggregation of the \textcolor{black}{ESI, QSI and CSI}, which leads to huge signaling \textcolor{black}{overhead}. In this section, we
shall first introduce a randomized base policy. Based on that, we shall propose a low complexity distributed deterministic policy using approximate dynamic programming\cite{Bertsekas:2007}. We shall show that the proposed solution has better performance than the randomized base policy.

\subsection{Randomized Base Policy}\label{subsec:indep-rand-policy}
We first introduce a randomized base policy  and discuss an important structural property of the equivalent Bellman equations in \eqref{eqn:Bellman-Q-factor} and \eqref{eqn:Bellman-postd}) under this base policy.

\begin{Def}[Randomized Base Policy] A randomized base policy is denoted as $\hat \Omega=(\hat \Omega_p,\hat \Omega_s)$.
The randomized base policy for BS-DTX control $\hat \Omega_p$ is given by a distribution on the action space of $\mathbf p$, i.e., $\boldsymbol{\mathcal P}$. The randomized base policy for user scheduling $\hat \Omega_s$ is given by a mapping from the CSI $\mathbf H$ to a probability distribution $\hat \Omega_s(\mathbf H)$ on the action space of $\mathbf s$, i.e.\textcolor{black}{,} $\boldsymbol{\mathcal S}$.~\hfill\QED\label{Def:QSI-ESI-indep-policy}
\end{Def}

Under a randomized base policy, the corresponding Q-factor and post-decision potential function have the following decomposition  structure.
\begin{Lem} [Decomposition under Randomized Base Policy] 

Given any randomized base policy $\hat\Omega$, the Q-factor $\hat{\mathbb Q} (\mathbf E,\mathbf Q, \mathbf p)$ and the potential function $\hat U(\widetilde{\mathbf E},\widetilde{\mathbf
Q})$ associated with the equivalent Bellman equations in  \eqref{eqn:Bellman-Q-factor} and  \eqref{eqn:Bellman-postd} can be expressed as:
$\hat{\mathbb Q}(\mathbf E, \mathbf Q,\mathbf p)=\sum_{b\in \mathcal B}\sum_{k\in \mathcal K_b}\hat{\mathbb Q}_k(E_b,Q_k,\mathbf p)$ and $\hat U(\widetilde{\mathbf E}, \widetilde{\mathbf Q})=\sum_{b\in \mathcal B}\sum_{k\in \mathcal K_b}\hat U_k(\widetilde E_b,\widetilde Q_k)$, where
\begin{align}
&\hat \theta_k+\hat{\mathbb Q}_k (E_b,Q_k, \mathbf p) \hspace{50mm}\forall E_b\in \mathcal E, Q_k \in \mathcal
Q , \mathbf p \in \boldsymbol{\mathcal P}\label{eqn:bellman-per-user-Q-factor}\\
=&\hat g_k(E_b, Q_k, p_b)+ \sum_{(E_b',Q_k') }\hat\Pr[ (
E_b',Q_k')|(E_b,Q_k),\mathbf p] \mathbb E^{\hat \Omega_p}[\mathbb Q_k ( E_b',Q_k',  \mathbf p')] \nonumber
\\
&\hat \theta_k+\hat U_k(\widetilde{
E}_b,\widetilde Q_k)\hspace{50mm} \forall   \widetilde{  E}_b\in \mathcal E,\widetilde Q_k
\in \mathcal Q \label{eqn:bellman-per-user-potential} \\
=&\sum_{(A^E_b, A^Q_k)}\Pr[
A^E_b]\Pr[A^Q_k]\left( \mathbb E^{\hat \Omega_p}[\hat g_k(E_b,Q_k, p_b)]+
\sum_{(\widetilde E_b',\widetilde Q_k') }\mathbb E^{\hat \Omega_p}\left[\Pr[(\widetilde E_b',\widetilde
Q_k')| (E_b, Q_k),\mathbf p]\right] \widetilde
V_k(\widetilde{\mathbf E}_n',\widetilde Q_k') \right)\nonumber
\end{align}
with $\hat g_k(E_b,Q_k, p_b)=\beta_k f(Q_k)+\gamma_b p_b\mathbf I[E_b=0]\mathbb E\left[\hat\Pr[s_k=1|\mathbf H]\right]$ and $ \hat \Pr\left[(E'_b,Q'_k)| (
E_b,Q_k), \mathbf p\right]=\mathbb E\left[\mathbb E^{\hat \Omega_s}\left[\Pr[(E'_b, Q'_k)|(E_b, Q_k,\mathbf H), p_b, s_k]\big|\mathbf H\right]\right]$.~\hfill\QED\label{Lem:decomposition}
\end{Lem}
\begin{proof} Please refer to Appendix C.
\end{proof}

\subsection{Low Complexity Delay-aware Distributed Solution}\label{subsec:low-complexity-sol}
Based  on the randomized base policy $\hat \Omega$, we shall obtain a low complexity distributed deterministic policy $\hat \Omega^*$ by Q-factor and potential function approximation. The solution is elaborated below.

\subsubsection{BS-DTX Control Policy Over a Longer Timescale} To reduce the complexity and to facilitate distributed
implementation, we approximate the BS-DTX control Q-factor $\mathbb
Q (\mathbf E,\mathbf Q, \mathbf p)$ in \eqref{eqn:Bellman-Q-factor}
by $\hat{\mathbb
Q} (\mathbf E,\mathbf Q, \mathbf p)$, i.e.,
\begin{align}
\mathbb Q (\mathbf E,\mathbf Q, \mathbf p) \thickapprox \hat{\mathbb
Q} (\mathbf E,\mathbf Q, \mathbf p)=\sum_{b\in \mathcal B}\sum_{k\in \mathcal K_b}\hat{\mathbb Q}_k(E_b,Q_k,\mathbf p)
\label{eqn:approximate-Q}
\end{align}
where $\hat{\mathbb Q}_k(E_b,Q_k)$ is given by the per-flow fixed point equation in \eqref{eqn:bellman-per-user-Q-factor}. The BSC determines the BS-DTX control based on the \textcolor{black}{aggregation of the ESI and QSI}
according to
\begin{align}
\hat{\mathbf p}^*(\mathbf E,\mathbf Q) =\arg \min_{\mathbf p \in \mathcal
P} \sum_{b\in \mathcal B}\sum_{k\in \mathcal K_b}\hat{\mathbb Q}_k(E_b,Q_k,\mathbf p)
\label{eqn:opt-pattern-selection-Q-factor-approx}
\end{align}

\begin{Rem} [Complexity of the BS-DTX Control] Under the linear Q-factor approximation in
\eqref{eqn:approximate-Q}, the complexity of obtaining the
BS-DTX control is reduced from  $\mathcal O \big((N_E+1)^B(N_Q+1)^K
2^B\big)$ to  $\mathcal O
\big((N_E+1)(N_Q+1)2^BK\big)$. To further reduce the complexity w.r.t. $B$, we can
partition the BSs into macro-groups with size $N_B$. The BS-DTX
control in \eqref{eqn:opt-pattern-selection-Q-factor-approx} can be
done for each of the $\frac{B}{N_B}$ macro-groups separately
\cite{BPCmultiCELL}. In practice, $N_B\ll B$ and hence, the
complexity becomes $\mathcal O ((N_E+1)(N_Q+1)2^{N_B}\frac{B}{N_B}K)$, which is
linear w.r.t. $B$.~\hfill\QED
\end{Rem}

\subsubsection{Distributed User Scheduling Policy at the CM Over a Shorter Timescale}

To reduce the complexity and to facilitate distributed
implementation of the user scheduling, we approximate the post-decision
state potential function $U(\widetilde{\mathbf
E},\widetilde{\mathbf Q})$ in \eqref{eqn:Bellman-postd} by $\hat U(\widetilde{\mathbf
E},\widetilde{\mathbf Q})$, i.e., 
\begin{align}
U(\widetilde{\mathbf E},\widetilde{\mathbf Q})
\thickapprox\hat U(\widetilde{\mathbf
E},\widetilde{\mathbf Q})= \sum_{b\in \mathcal B}\sum_{k\in \mathcal K_b}\hat U_k(E_b,Q_k)
\label{eqn:approximate-V}
\end{align}
where $\hat U_k(\widetilde E_b,\widetilde
Q_k)$  is given by the per-flow fixed point equation in \eqref{eqn:bellman-per-user-potential}.  
Substituting the approximation in \eqref{eqn:approximate-V} into the optimal user scheduling in \eqref{eqn:opt-user-selection-postd-potential}, the user scheduling solution under the approximation is summarized below.

\begin{Lem}[Distributed User Scheduling] Under the linear potential function approximation in
\eqref{eqn:approximate-V}, the distributed user scheduling action
$\hat{\mathbf s}_n^*$ of the $n$-th cluster based on the \textcolor{black}{per-cluster ESI, QSI and
CSI} under  $\hat{\mathbf p}^*(\mathbf E, \mathbf Q)$ obtained by
\eqref{eqn:opt-pattern-selection-Q-factor-approx} is given by
\begin{align}
&\hat{\mathbf s}_n^*(\mathbf E_n, \mathbf Q_n, \mathbf H_n),\hspace{50mm}
\forall \mathbf E_n \in \boldsymbol{\mathcal E}_n, \mathbf Q_n \in
\boldsymbol{\mathcal Q}_n, \mathbf H_n \in \boldsymbol{\mathcal
H}_n, \forall
n \label{eqn:decentralized-user-selection}\\
=&\arg \max_{\mathbf s_n \in \boldsymbol{\mathcal S}_n(\hat{\mathbf
p}_n^*)} \sum_{k \in \mathcal K_n} s_k \Pr[ \rho_k
(\mathbf H_n, \hat{\mathbf p}^*, \mathbf s_n)\geq \delta_k]\left(\hat
U_k([E_b-\hat{p}_b^*]^+, Q_k)-\hat U_k([
E_b-\hat{p}_b^*]^+, [Q_k-1]^+)\right)\nonumber
\end{align}
where $\boldsymbol{\mathcal S}_n(\mathbf p_n)\triangleq \{\mathbf
s_n\in \mathcal S^{\sum_{b\in \mathcal B_n}K_b}: \sum_{k \in \mathcal K_b}s_k=p_b, b\in \mathcal B_n\}$
denotes the feasible user scheduling action space of cluster $n$
under the BS-DTX control action $\mathbf p_n$.
\label{Lem:decentralized-user-selection}~\hfill\QED
\end{Lem}
\begin{proof}
Please refer to Appendix D.
\end{proof}

\begin{Rem}[Complexity of the User Scheduling] The user scheduling action in
\eqref{eqn:decentralized-user-selection} is a function of  \textcolor{black}{the per-cluster ESI,
QSI and CSI}, and is computed locally at the $n$-th CM. Under the
linear potential function approximation in
\eqref{eqn:approximate-V}, the complexity of  
user scheduling is reduced from  $\mathcal O
\big((N_E+1)^B(N_Q+1)^K \big)$ to $\mathcal O
\big((N_E+1)(N_Q+1)K\big)$.~\hfill\QED
\end{Rem}

\subsection{Performance of Low Complexity Delay-aware Distributed Solution}

The key motivation of the linear approximatios of the Q-function and potential function in \eqref{eqn:approximate-Q} and \eqref{eqn:approximate-V} is to facilitate distributed control. The following theorem shows that the proposed distributed policy always achieves better performance than the randomized base policy.

\begin{Thm} [Performance Improvement] If $\Pr[(\mathbf E',\mathbf Q')| (\mathbf E,\mathbf Q), (\mathbf p,\mathbf s) ]\neq \Pr[(\mathbf E',\mathbf
Q')| (\mathbf E,\mathbf Q), (\mathbf p',\mathbf s')]$ for any  $(\mathbf p,\mathbf s)\neq(\mathbf p',\mathbf s')$ and $(\mathbf E,\mathbf Q)\in \boldsymbol{\mathcal E}\times \boldsymbol{\mathcal Q}$,  then we have $\hat \theta^*(\mathbf E, \mathbf Q)<\hat \theta$ for all $(\mathbf E, \mathbf Q)\in \boldsymbol{\mathcal E}\times\boldsymbol{\mathcal Q}$,  where $\hat \theta^*(\mathbf E, \mathbf Q)$ is the average cost under the proposed solution starting from state $(\mathbf E, \mathbf Q)$ and $\hat \theta$ is the average cost under any  randomized base policy, respectively. ~\hfill\QED\label{Thm:performance-comp}
\end{Thm}
\begin{proof} Please refer to Appendix E.
\end{proof}

\section{Distributed Online Learning via Stochastic Approximation}\label{Sec:distributed-learning}

Observe that the BS-DTX control and the user scheduling in
\eqref{eqn:opt-pattern-selection-Q-factor-approx} and
\eqref{eqn:decentralized-user-selection} require the knowledge of
$\{\hat{\mathbb Q}_k (E_b, Q_k, \mathbf p)\}$ and $\{\hat
U_k(\widetilde{E}_b,\widetilde Q_k)\}$, 
respectively, which are defined in the fixed point equations \textcolor{black}{in} 
\eqref{eqn:bellman-per-user-Q-factor} and
\eqref{eqn:bellman-per-user-potential}, respectively. However,
solving these fixed point equations \textcolor{black}{is}  also quite challenging.  In
this section, we shall propose an online distributed stochastic
learning  \cite{SA:Q_learning} algorithm to estimate $\{\hat{\mathbb Q}_k (E_b, Q_k, \mathbf p)\}$ and $\{\hat
U_k(\widetilde{E}_b,\widetilde Q_k)\}$
 using \textcolor{black}{the}  per-cluster system state information
only. We shall prove that the proposed  distributed  algorithm
converges almost surely to the fixed point solutions.

\subsection{Distributed  Online Learning for $\{\hat{\mathbb Q}_k (E_b, Q_k, \mathbf p)\}$ and $\{\hat U_k(\widetilde{E}_b,\widetilde Q_k)\}$}

Since the statistics of $\mathbf A^Q(t)$ and $\mathbf A^E(t)$ are unknown to the controller, instead of computing $\{\hat{\mathbb Q}_k (E_b, Q_k, \mathbf p)\}$ and $\{\hat
U_k(\widetilde{E}_b,\widetilde Q_k)\}$ of a chosen $\hat \Omega$
offline, we shall estimate them distributively at each BS based on
the instantaneous observations.
\begin{Alg}(\emph{Online Per-User Q-factor and Potential
Function Learning Algorithm})
\begin{itemize}
\item {\bf Step 1 [Initialization at the BSs]}: Set $t=0$. Each BS $b$
initializes $\{\hat{\mathbb Q}_k^0 (E_b, Q_k, \mathbf p)\}$ and $\{\hat
U_k^0(\widetilde{E}_b,\widetilde Q_k)\}$ for all $k\in \mathcal K_b$.

\item {\bf Step 2 [\textcolor{black}{BS-DTX} Control at the BSC]}: At the beginning of the $t$-th
slot, each BS $b$ reports \textcolor{black}{$\left\{\sum_{k\in \mathcal K_b}\hat{\mathbb Q}^t_k\left(E_b(t),Q_k(t),\mathbf p\right):\mathbf p\in \boldsymbol{\mathcal P}\right\}$} to the BSC. The BSC
determines BS-DTX control $\hat{\mathbf p}^*(t)\triangleq\hat{\mathbf
p}^*(\mathbf E(t),\mathbf Q(t))$ according to
\eqref{eqn:opt-pattern-selection-Q-factor-approx} and  broadcasts $\hat{\mathbf
p}^*(t)$ to all the CMs. Each CM $n$ informs $\hat p_b^*(t)$ to each BS
$b\in \mathcal B_n$. Each BS $b$ determines its renewable and grid 
power allocations, i.e., $\hat p_b^{E*}(t)=\hat p_b^*(t)\mathbf I[E_b(t)>0]$
and $\hat p_b^{G*}(t)=\hat p_b^*(t)\mathbf I[E_b(t)=0]$, respectively.

\item {\bf Step 3 [User Scheduling at the CMs]}: \textcolor{black}{Each BS $b$ reports $\left\{\hat{U}^t_k\left(E_b(t),Q_k(t)\right):k\in \mathcal K_b\right\}$ to its CM.}  Each CM $n$ determines user selection $\hat{\mathbf
s}_n^*(t)\triangleq\hat{\mathbf s}_n^*(\mathbf E_n(t),\mathbf Q_n(t), \mathbf
H_n(t))$ according to \eqref{eqn:decentralized-user-selection} under
given  BS-DTX control $\hat{\mathbf p}^*(t)$.

\item {\bf Step 4 [Per-flow Q-factor and Potential Function Update at the BSs]}: Based on the current  observations  $A^E_b(t)$ and $A^Q_k(t)$ ($k\in \mathcal K_b$), each BS $b$ updates the per-flow Q-factor and
potential function for the MSs in its cell according to
\eqref{eqn:update-Q} and \eqref{eqn:update-postd-V} for all $k\in \mathcal K_b$.
\begin{align}
&\hat{\mathbb Q}^{t+1}_k (E_b,Q_k, \mathbf p) \hspace{50mm} \forall  E_b\in \mathcal
E, Q_k \in \mathcal Q , \mathbf p \in \boldsymbol{\mathcal
P}\label{eqn:update-Q}\\
=& \hat{\mathbb Q}^t_k
(E_b,Q_k, \mathbf p)+\epsilon_t \left[F_k(\hat{\mathbb Q}^{t}_k, E_b,Q_k, \mathbf p)-F_k(\hat{\mathbb Q}^{t}_k, E_b^I,Q_k^I, \mathbf p^I)-\mathbb Q^t_k ( E_b,Q_k, \mathbf p)\right]
\nonumber\\
&\hat U^{t+1}_k (\widetilde{E}_b,\widetilde Q_k)\hspace{55mm} \forall
\widetilde{E}_b\in \mathcal  E, \widetilde Q_k
\in \mathcal Q \label{eqn:update-postd-V}\\
=& \hat
U^t_k (\widetilde{E}_b,\widetilde Q_k)+
\epsilon_t\left[T_k(\hat{\mathbf  U}^{t}_k, \widetilde E_b,\widetilde Q_k)-T_k(\hat{\mathbf  U}^{t}_k, \widetilde E_b^I,\widetilde Q_k^I) -\hat U^t_k (\widetilde{E}_b,\widetilde
Q_k)\right] \nonumber
\end{align}
where 
\begin{align}
F_k(\hat{\mathbb Q}^{t}_k, E_b,Q_k, \mathbf p)=&\hat g_k(E_b,Q_k, p_b)+ \sum_{(E_b',Q_k') }\hat\Pr[ (
E_b',Q_k')|(E_b,Q_k),\mathbf p]\nonumber\\
&\times\mathbb E^{\hat \Omega_p}\left[\mathbb Q_k^t ( \min\{E_b'+A^E_b(t),N_E\},\min\{Q_k'+A^Q_k(t),N_Q\},  \mathbf p')\right]\label{eqn:update-F-Q}
\end{align}
\begin{align}
&T^t_k(\hat{\mathbf  U}^{t}_k, \widetilde E_b,\widetilde Q_k)=\mathbb E^{\hat \Omega_p}\left[\hat g_k\left(\min\{\widetilde E_b+A^E_b(t),N_E\},\min\{\widetilde Q_k+A^Q_k(t),N_Q\}, p_b\right)\right]\nonumber\\
&+\sum_{(\widetilde E_b',\widetilde Q_k') }\mathbb E^{\hat \Omega_p}\left[\hat \Pr\left[(\widetilde E_b',\widetilde
Q_k')|\left(\min\{\widetilde E_b+A^E_b(t),N_E\},\min\{\widetilde Q_k+A^Q_k(t),N_Q\}\right),\mathbf p\right]\right] \hat
U_k^t(\widetilde E_b',\widetilde Q_k')\label{eqn:update-T-J}
\end{align}
$ E_b'=[E_b-p_b]^+$, $Q_k'=\left[Q_k-\mathbf I[\rho_k
(\mathbf H_n, \mathbf p, \mathbf s_n)\geq \delta_k]\right]^+$, $\widetilde E_b'=\left[\min\{\widetilde E_b+A^E_b(t),N_E\}-p_b\right]^+$, $\widetilde Q_k'=\left[\min\{\widetilde Q_k+A^Q_k(t),N_Q\}-\mathbf I[\rho_k
(\mathbf H_n, \mathbf p, \mathbf s_n)\geq \delta_k]\right]^+$. 
$\mathbf p^I$ is the
reference BS-DTX control action and  $E_b^I$, $Q_k^I$, $\widetilde E_b^I$, $\widetilde Q_k^I$   are the reference states\footnote{The reference action and  states
 are used to bootstrap the online learning
algorithms \cite{SA:ODE:twoscale} for
\eqref{eqn:update-Q} and \eqref{eqn:update-postd-V} respectively. 
Without loss of generality, we set $E^I_b=0$,
$Q^I_k=0$, $\mathbf p^I=\{p_b^I=1:b\in \mathcal B\}$, $\widetilde E^I_b= 0$ and  $\widetilde Q^I_k=0$.
} for the Q-factor update in \eqref{eqn:update-Q} and the
potential function update in \eqref{eqn:update-postd-V},
respectively. $\{\epsilon_t\}$ are diminishing positive step
size sequences satisfying the following conditions:
$\epsilon_t\geq 0, \sum_t\epsilon_t=\infty, \ \sum_t\epsilon^2_t<\infty$.
\end{itemize}
\label{Alg:general_alg}
\end{Alg}

\subsection{\textcolor{black}{Performance} of the  Distributed  Learning Algorithm}\label{subsec:general_convergence_proof}

The convergence of Algorithm \ref{Alg:general_alg} is summarized
below.
\begin{Lem}[Convergence of Algorithm \ref{Alg:general_alg}] The iterative updates of the per-flow Q-factor and the per-flow potential function in
\eqref{eqn:update-Q}  and \eqref{eqn:update-postd-V} converge
almost surely, i.e., $\lim_{t\to \infty}\boldsymbol{ \mathbb
Q}^t_k=\hat{\mathbb Q}^{\infty}_k$ a.s. and $\lim_{t\to
\infty}\hat {\mathbf U}^t_k=\hat {\mathbf U}^{\infty}_k$ a.s. 
($\forall k \in \mathcal K$), where  $\hat{\boldsymbol{ \mathbb Q}}^{\infty}_k$
and  $\hat {\mathbf U}^{\infty}_k$ are the  solutions of the fixed point
equations in \eqref{eqn:bellman-per-user-Q-factor} and
\eqref{eqn:bellman-per-user-potential}, respectively.~\hfill\QED
\label{Lem:convergence-update}
\end{Lem}
\begin{proof}
Please refer to Appendix F.
\end{proof}

\begin{Rem} [Signaling Requirement of Distributed Two-Timescale Algorithm \ref{Alg:general_alg}]$\quad$

\textcolor{black}{$\bullet$ {\bf Signaling requirement over a short timescale (per slot)}: Each BS needs to collect the local CSI over the radio interface. The BSs within a cluster also need to report the local CSI to its CM. Yet, the signaling loading and the latency requirement for this part is in fact similar to the existing HSDPA and LTE systems. }

\textcolor{black}{$\bullet$ {\bf Signaling requirement through the backhaul over a long timescale (in convergent stage)}: Each BS needs to report the Q-factors of the (updated) local QSI  to the BSC  (for the BS-DTX control) as well as the potential functions  of the (updated) local QSI to the CM (for the user scheduling within a cluster). These signaling exchanges are over the high-speed backhaul and over a longer timescale (not on a slot by slot basis). The latency of signaling over backhaul (typically less than 10ms) is negligible.}~\hfill\QED
\end{Rem}

\section{Stability Analysis}\label{Sec:stablity-analysis}

In this section, we shall analyze the stability conditions for the
data queues in the coordinated MIMO networks  with infinite data
buffer size ($N_Q=\infty$) and finite energy storage size ($N_E<\infty$), and discuss
various design insights. We have the
following assumption on the BS and MS distributions.

\begin{Asump} [BS and MS Distributions] The location of the BSs follows a homogeneous {\em Poisson Point
Process} (PPP) $\Phi$ of density $\lambda$ and the location of the
MSs follows some independent stationary point process in the
Euclidean plane\cite{JeffSGCellularNet10,SGbookHaenggi08Now}. Each
MS is associated with the closest BS, i.e., the MSs in the Voronoi
cell of a BS are associated with it.~\hfill\QED\label{Asump:PPP}
\end{Asump}
To simplify the analysis, we consider a homogeneous network with $K_b=1$, 
$P_b=P$ $\forall b\in \mathcal B$ and $\delta_k=\delta$ $\forall
k\in \mathcal K$. In addition,  we assume the CSI follows
complex Gaussian fading and the long-term path gain follows standard
power law $L_{k,b}=r_{k,b}^{-\alpha}$, where $r_{k,b}$ is the
distance between BS $b$ and MS $k$ and  $\alpha>2$ is the path loss
exponent. Furthermore, the renewable energy and bursty data
arrivals under Assumptions  
\ref{Asump:general_A-Q} and \ref{Asump:general_A-E} are specialized to Bernoulli processes,
i.e., $A^Q_k(t), A^E_b(t) \in \{0,1\}$, \textcolor{black}{$\mathbb E[A^Q_k(t)]=\lambda^Q<
1$ } and $\mathbb
E[A^E_b(t)]=\lambda^E< 1$ for all  $k\in \mathcal K$ and $ b\in \mathcal B$. We consider the
following  randomized BS-DTX 
policy.
\begin{Def}[Randomized BS-DTX control Policy]
At each slot $t$,   each BS $b\in \mathcal B$ is active with probability $p_{tx}>0$, i.e., $\Pr[p_b(t)=1]=p_{tx}$, if $\sum_{k \in \mathcal K_b} Q_k(t)>0$; $p_b(t)=0$ otherwise.~\hfill\QED\label{Def:rand-policy}
\end{Def}
In the following, we shall analyze the sufficient conditions for the
queue stability (i.e., $Q_k(t)$ having a steady state limiting
distribution for $t\to \infty$
\cite{Szpankowskistabilityconditions:93}) under the randomized policy in Definition \ref{Def:rand-policy}  of a randomly chosen user.

\subsection{Stability Analysis for Systems without BS Coordination ($N_t=1$)}\label{subsec:analysis-no-coop}

In this case, we consider no cooperation among BSs ($N_t=1$). Using  stochastic geometry \cite{SGbookHaenggi08Now} 
and the technique of parallel dominant queues
\cite{EphremidesEH11,Raointeractingqueue:88}, the following lemma
summarizes the sufficient condition for the queue stability of a
randomly chosen MS at a distance $r_1$ from its BS\footnote{When $N_E\to\infty$ and $p_{tx}=1$, the result in \eqref{eqn:stability-no-cooperation} reduces to the coverage probability for cellular networks without BS coordination obtained in \cite{JeffSGCellularNet10}. }. 

\begin{Lem} [Sufficient Condition for Queue Stability without BS Coordination] The data queue of
a  randomly chosen MS  
is stable if
\begin{align}
\lambda^Q< 
 p_{tx} \exp\left(-C_1p_{tx}\lambda -\frac{N_0}{P}\delta r_1^{\alpha}\right)\triangleq
\lambda^Q_{\max}(p_{tx},N_t)\label{eqn:stability-no-cooperation}
\end{align}
In addition, $\lambda^Q_{\max}(p_{tx},N_t)$ corresponds to the maximum average \textcolor{black}{grid power} cost per BS 
$\overline{p^{G}_{\max}}(p_{tx},N_E)=\big(1-f(p_{tx},N_E)\big)p_{tx}$. $N_t=1$, $C_1=\frac{2}{\alpha-2}\pi  r_1^2\delta $ 
and
\begin{align}
f(p_{tx},N_E) =
\begin{cases}
\frac{(\lambda^E/p_{tx})\big(1-(\lambda^E/p_{tx})^{N_E}\big)}{1-(\lambda^E/p_{tx})^{N_E+1}},&
\lambda^E\neq p_{tx}\\
\frac{N_E}{N_E+1},& \lambda^E= p_{tx}
\end{cases}
\label{eqn:stability-f}
\end{align}
~\hfill\QED\label{Lem:stability-stability-no-cooperation}
\end{Lem}
\begin{proof} Please refer to Appendix G.
\end{proof}

\begin{Rem} [Interpretation of Lemma \ref{Lem:stability-stability-no-cooperation}] $f(p_{tx},N_E)$ can be interpreted as the probability that a
energy queue is non-empty in a parallel dominant network\footnote{In the parallel dominant network, dummy packets are
transmitted if a data queue is empty. Thus, the BS-sDTX controls are decoupled from the data queues, i.e.,
independent of the QSI, and hence, the renewable and \textcolor{black}{grid power}
consumptions are symmetric across all the BSs.}, i.e.,
$f(p_{tx},N_E)=\Pr[E_b>0]$. It can be easily verified from \eqref{eqn:stability-f} that
$f(p_{tx},N_E)$ increases as $N_E$ increases and $\lim_{N_E\to
\infty}f(p_{tx},N_E)=\min\{\frac{\lambda^E}{p_{tx}},1\}\triangleq
f(p_{tx},\infty)$, which corresponds to the case with infinite
energy storage size. In addition, $\overline{p^{G}_{\max}}(p_{tx},N_E)=\Pr[p_b^{G}=1]$.~\hfill\QED
\end{Rem}

\subsection{Stability Analysis for Systems with BS Coordination
($N_t>1$)}\label{subsec:analysis-coop}

In this part, we \textcolor{black}{extend} the analysis to the case with BS
coordination ($N_t>1$). For a randomly chosen MS, the interference comes from the active BSs outside its cluster.  Hence, we need to consider the distribution of the coordination clusters and the associated analysis is more challenging compared with the case without BS coordination  ($N_t=1$) \cite{JeffSGCellularNet10}.

\begin{Lem} [Sufficient Condition for Queue Stability with BS Coordination] For $N_t>\frac{\alpha}{2}$,\footnote{Note that  $N_t>\frac{\alpha}{2}$ implies $N_t>1$ for most of the cases we are interested in, as wu usually have $2<\alpha <4$  in practical systems. } the data queue of
a  randomly chosen MS in the coordinated MIMO network can be stabilized if
\begin{align}
\lambda^Q<
p_{tx} \exp\left(-C_{N_t}p_{tx}\lambda -\frac{N_0}{P}\delta r_1^{\alpha}\right)\triangleq
\lambda^Q_{\max}(p_{tx}, N_t)
\label{eqn:stability-cooperation}
\end{align}
In addition, $\lambda^Q_{\max}(p_{tx},N_t)$ corresponds to the maximum average \textcolor{black}{grid power} cost per BS 
$\overline{p^{G}_{\max}}(p_{tx},N_E)=\big(1-f(p_{tx},N_E)\big)p_{tx}$.   
$C_{N_t}=\frac{2}{\alpha-2}\pi  r_1^2\delta \min\{1,r_1^{\alpha-2}\frac{(\lambda \pi)^{-1+\frac{\alpha}{2}}\Gamma(N_t-\frac{\alpha}{2})}{\Gamma(N_t-1)}\}=\mathcal O(N_t^{1-\frac{\alpha}{2}})$ as $N_t \to \infty$.~\hfill\QED\label{Lem:stability-stability-cooperation}
\end{Lem}
\begin{proof}
Please refer to Appendix H.
\end{proof}

\subsection{Optimization of Randomized Policy}

We are interested in maximizing
$\lambda^Q_{\max}(p_{tx},N_t)$  under \textcolor{black}{grid power} constraint $P^{G}$ w.r.t. the parameter $p_{tx}$ in the randomized control policy  for any given $N_E>0$ and $N_t\geq 1$.
Specifically, we have
\begin{align}
p_{tx}^*(N_E,N_t)=\arg\max_{p_{tx}\in[0,1]}\ & \lambda^Q_{\max}(p_{tx},N_t)
\label{eqn:max-no-coorp}\\
s.t. \quad &\overline{p^{G}_{\max}}(p_{tx},N_E)\leq
P^{G}\nonumber
\end{align}
Let  $\lambda^{Q*}_{\max}(N_E,N_t)=\lambda^Q_{\max}(p_{tx}^*,N_t)
$ denote the optimal value of the
optimization problem in \eqref{eqn:max-no-coorp}. Let $x^*(N_E)$ denote the solution to $\overline{p^{G}_{\max}}(x,N_E)=
P^{G}$ for any \textcolor{black}{given} $N_E>0$. The following theorem summarizes the optimal solution.

\begin{Thm} [Optimization Solution for Queue Stability]
$p_{tx}^*(N_E,N_t)=\min\{x^*(N_E),1, \frac{1}{C_{N_t} \lambda}\}$.  For any given $N_t\geq1$, $\lambda^{Q*}_{\max}(N_E,N_t)$ is strictly increasing in $N_E$ if $x^*(N_E)<\min\{1, \frac{1}{C_{N_t} \lambda}\}$ and is a constant for all $N_E$ if $x^*(N_E)\geq \min\{1, \frac{1}{C_{N_t} \lambda}\}$. For any given $N_E>0$, $\lambda^{Q*}_{\max}(N_E,N_t)$ is strictly increasing in $N_t$.~\hfill\QED\label{Thm:prop-stability-stability-cooperation}
\end{Thm}
\begin{proof}
Please refer to Appendix I.
\end{proof}

\section{Results and Discussions}

In this section, we shall discuss the design insights from the analytical results in Section \ref{Sec:stablity-analysis}. We also compare the delay performance gain of the proposed delay-aware low complexity distributed scheme in Section \ref{sec:general_decentral_algo} and Section \ref{Sec:distributed-learning} with the following two baseline schemes using simulation.

$\bullet$ {\bf Baseline 1 [CSI-based Single Cell Scheme]:} Baseline
1 refers to the  randomized BS-DTX control and CSI-based user scheduling
without BS coordination. Each multi-antenna BS uses maximal ratio combining (MRC) and  selects  one MS with the maximum successful packet transmission probability based on the observed \textcolor{black}{local} CSI.

$\bullet$  {\bf Baseline 2 [CSI-based Clustered Coordinated MIMO Scheme]:} Baseline 2
refers to the randomized BS-DTX control and CSI-based clustered coordinated MIMO with the same coordinated beamforming as the proposed scheme.  Each CM determines the user scheduling to maximize the sum
successful packet transmission probability of each cluster based on the observed \textcolor{black}{per-cluster CSI}.

In the simulation, we consider a cellular system with 19 BSs, each has a coverage of 500m
and 2 mobiles per cell, which distribute uniformly
in the cell-edge with range [400m, 500m] from the BS. We apply the
Urban Macrocell Model in 3GPP \cite{3GPP} with path loss model given
by $PL=34.5+35\log_{10}(r)$, where $r$ (in m) is the distance from
the transmitter to the receiver. Each element of  $\mathbf h_{k,b}$ is $\mathcal{CN}(0,1)$. The
total bandwidth is 1MHz. The BS transmit power is $P_b=35$ dBm for all $b\in \mathcal B$, the threshold is $\delta_k=0.5$, and $\beta_k=1$ for all $k\in \mathcal K$. $\gamma_b$ is the same for all $b\in \mathcal B$. We consider Bernoulli arrival processes for the renewable energy and busty data arrivals. The maximum buffer size $N_Q=15$ pcks.


\subsection{Effect of BS Coordination} From Theorem
\ref{Thm:prop-stability-stability-cooperation},  we can see that for
any given $N_E>0$, the loading supported by the energy harvesting
\textcolor{black}{system} $\lambda^{Q*}_{\max}(N_E, N_t)$  increases as $N_t$ increases. Intuitively, the gain comes from  BS coordination. Fig. \ref{Fig:delay_vs_pwr}  illustrates the average delay versus the average transmit power cost for different number of transmit antennas $N_t$. It can be observed that the average delay of Baseline 2 and the proposed scheme decreases as $N_t$ increases. \textcolor{black}{This demonstrates}  that BS coordination improves the delay performance.

\subsection{Effect of Energy Buffer Size} 

From Theorem \ref{Thm:prop-stability-stability-cooperation}, we can see that for
any given $N_t\geq 1$, the loading supported by the energy harvesting
\textcolor{black}{system} $\lambda^{Q*}_{\max}(N_E, N_t)$  increases in
  $N_E$. Specifically, when $x^*(N_E)<\min\{1, \frac{1}{C_{N_t} \lambda}\}$, $\lambda^{Q*}_{\max}(N_E, N_t)$ 
increases as $N_E$ increases. The intuition is that the above
condition corresponds to  the power-limited region. By increasing
$N_E$,  more renewable energy can be accumulated due to less
renewable energy loss when the energy storage is full, and hence more
traffic loading can be supported. However, when $x^*(N_E)\geq \min\{1, \frac{1}{C_{N_t} \lambda}\}$,  $\lambda^{Q*}_{\max}(N_E,N_t)$ is
constant for all $N_E>0$. The intuition is that the above condition 
corresponds to  the interference-limited region, in which the traffic
loading supported cannot be  increased by accumulating more
renewable energy through increasing $N_E$. 
Fig.  \ref{Fig:delay_vs_NE}  illustrates the average delay versus the energy storage size $N_E$ at  average transmit \textcolor{black}{grid power} 15 dBm.  It can be observed that the average delay decreases as the energy storage size increases for all the schemes.

\subsection{Performance of the Proposed Scheme}

Fig. \ref{Fig:delay_vs_arrival} illustrates the average delay versus
per-flow loading (average arrival rate $\lambda_k$). The average delay of all the schemes increases
as the loading increases. The proposed scheme also achieves
significant gain over the baselines across a wide range of input
loading. Fig. \ref{Fig:convergence} illustrates the convergence property of the
proposed distributed online learning algorithm for
estimating the per-flow potential function and the per-flow
Q-factor. It can be observed that the proposed distributed learning algorithm converges quite
fast. Furthermore, the average delay at the
the 500-th scheduling slot is 4.1853 pcks, which is much smaller
than the other baselines.

\section{Summary}

In this paper, we propose a {\em two-timescale} delay-optimal BS-DTX control and user scheduling  for energy
harvesting downlink coordinated MIMO networks.  We show that the
two-timescale delay-optimal control problem can be modeled as a
POMDP and derive the optimal centralized control. To reduce the
complexity and facilitate the distributed implementation, we obtain a
distributed solution with the BS-DTX control at the BSC based on \textcolor{black}{the aggregation of the 
ESI and QSI} and the user scheduling at each CM based on \textcolor{black}{the per-cluster 
ESI, QSI and CSI} with guaranteed delay performance. We prove the almost-sure
convergence of the proposed distributed two-timescale algorithm.
Furthermore, we analyze the stability conditions for the data
queues in coordinated MIMO networks and discuss various design
insights.
\begin{appendix}

\section*{Appendix A: Proof of Lemma \ref{Lem:opt-re-ac-pwr-control}}

We shall prove Lemma \ref{Lem:opt-re-ac-pwr-control} using sample
path arguments. Let $\{\mathbf A^E(\omega,t)\}$, $\{\mathbf
A^Q(\omega,t)\}$ and $\{\mathbf H(\omega,t)\}$ be a given sample
path (i.e., $\omega$) of energy arrivals, packet arrivals and CSI
states. Let $\{\mathbf p(t)\}$ and $\{\mathbf s(t)\}$ be any given
sequences of feasible BS-DTX control actions and user scheduling
actions. Note that for given $\{\mathbf p(t)\}$ and $\{\mathbf
s(t)\}$, the trajectory of QSI $\{\mathbf Q(\omega,t)\}$ is uniquely
determined. Let $\{\mathbf p^E(\omega,t)\}$ and $\{\mathbf
p^{G}(\omega,t)\}$ be the sequences of the renewable power DTX
control actions and \textcolor{black}{grid power} DTX control actions satisfying the
structure in Lemma \ref{Lem:opt-re-ac-pwr-control} for the given
$\{\mathbf p(t)\}$, i.e.,
$p_b^E(\omega,t)=p_b(t)\mathbf I[E_b(\omega, t)>0]$ and
$p_b^{G}(\omega, t)=p_b(t)\mathbf I[E_b(\omega, t)=0]$), where
$\{\mathbf E(\omega,t)\}$ is the trajectory of ESI associated with
$\{\mathbf p^E(\omega,t)\}$. Let $\{\mathbf p^E{'}(\omega,t)\}$ and
$\{\mathbf p^{G}{'}(\omega,t)\}$ be any other sequences of feasible
renewable power DTX control actions and \textcolor{black}{grid power} DTX control actions
conditioned on $\{\mathbf p(\omega,t)\}$, i.e.,
$p_b^E{'}(\omega,t)+p_b^{G}{'}(\omega,t)=p_b(\omega,t)$, and
$\{\mathbf E'(\omega,t)\}$ be the trajectory of ESI associated with
$\{\mathbf p^E{'}(\omega,t)\}$.

In the following, for each $b\in \mathcal B$, we shall show that for $\mathbf E'(\omega,1)=\mathbf E(\omega,1)$, we have $\sum_{t=1}^T
p_b^{G}(\omega,t)\leq \sum_{t=1}^T p_b^{G}{'}(\omega,t)$. Let 
$\Delta p^{G}_b(\omega,T)\triangleq \sum_{t=1}^{T-1}\big(
p_b^{G}{'}(\omega,t)- p_b^{G}(\omega,t)\big)$ $\forall T\geq 2$ and $\Delta p^{G}_b(\omega,1)=0$. Then, we have $\Delta
p^{G}_b(\omega,t+1)=\Delta p^{G}_b(\omega,t)+\big(
p_b^{G}{'}(\omega, t)- p_b^{G}(\omega,t)\big)$ for all $t\geq 1$.
We shall prove $ E_b(\omega,t)+\Delta p^{G}_b(t)\geq
E_b'(\omega,t)$ and $\Delta p^{G}_b(\omega,t)\geq 0$ for all $t\geq
1$ by induction. (In the following proof, we omit $\omega$ for
notation simplicity.)
\begin{itemize}
\item Consider $t=1$. Since $E'_b(1)=
E_b(1)$ and $\Delta p^{G}_b(1)=0$ by the initial condition, we have
$ E_b(1)+\Delta p^{G}_b(1)\geq E_b'(1)$ and $\Delta p^{G}_b(1)\geq
0$.

\item For some $t\geq 1$, assume $
E_b( t)+\Delta p^{G}_b(t)\geq E_b'( t)$ and $\Delta p^{G}_b(t)\geq
0$. $E_b'(t+1)=\min\{E_b'(t)-p^E{'}(t)+A^E_b(t),N_E\}$. We shall
show the conclusions hold for $t+1$ by considering the following
three cases. (1) When $E_b(t)>0$, we have $p^E_b(t)=p_b(t)$ and
$p^{G}_b(t)=0$. Thus, we have
$E_b(t+1)=\min\{E_b(t)-p(t)+A^E_b(t),N_E\}$ and $\Delta
p^{G}_b(t+1)=\Delta p^{G}_b(t)+p^{G}_b{'}(t)-0\geq0$. In
addition, since $p^{G}_b{'}(t)=p_b(t)-p^E_b{'}(t)$, we have
$E_b(t+1)+\Delta p^{G}_b(t+1)=\min\{E_b(t)-p_b(t)+A^E_b(t)+\Delta
p^{G}_b(t)+p_b(t)-p^E_b{'}(t),N_E+\Delta p^{G}_b(t+1)\}\geq
E_b'(t+1)$. (2)  When $E_b(t)=0$ and $E_b'(t)\geq 1$, which implies
$\Delta p^{G}_b(t)\geq 1$, we have $p^E_b(t)=0$ and
$p^{G}_b(t)=p_b(t)$. Thus, we have
$E_b(t+1)=\min\{E_b(t)+A^E_b(t),N_E\}$  and $\Delta
p^{G}_b(t+1)=\Delta p^{G}_b(t)+p^{G}_b{'}(t)-p_b(t)=\Delta
p^{G}_b(t)-p^E_b{'}(t)\geq 1-1=0$, and hence, we have
$E_b(t+1)+\Delta p^{G}_b(t+1)=\min\{E_b(t) +A^E_b(t)+\Delta
p^{G}_b(t)-p^E_b{'}(t),N_E+\Delta p^{G}_b(t+1)\}\geq E_b'(t+1)$.
(3) When $E_b(t)=0$ and $E_b'(t)= 0$, we have
$p^E_b(t)=p^E_b{'}(t)=0$ and $p^{G}_b(t)=p^{G}_b{'}(t)=p_b(t)$.
Thus, we have $E_b(t+1)=\min\{0+A^E_b(t),N_E\}$,
$E_b'(t+1)=\min\{0+A^E_b(t),N_E\}$ and $\Delta p^{G}_b(t+1)=\Delta
p^{G}_b(t)+0\geq0$, and hence, $E_b(t+1)+\Delta p^{G}_b(t+1)=
E_b'(t+1)+\Delta p^{G}_b(t+1)\geq  E_b'(t+1)$.  
\end{itemize}
Therefore, by induction, we can show $\Delta p^{G}_b(\omega,t)\geq
0$ for all $t$. Since the average delay costs per stage are the
same, we have $$\frac{1}{T}\sum_{t=1}^T \left(\sum_k \beta_k
f(Q_k(\omega,t))+\sum_b \gamma_b  p^{G}_b(\omega,t)\right)\leq
\frac{1}{T}\sum_{t=1}^T \left(\sum_k \beta_k f(Q_k(\omega,t))+\sum_b
\gamma_b  p^{G}_b{'}(\omega,t)\right)$$ for any given $\{\mathbf
p(t)\}$ and $\{\mathbf s(t)\}$ and $T$. By taking expectations over
all sample paths,  $\limsup$ and optimizations over BS-DTX control
and user selection policy space, we have
$\min_{\pi}J_{\pi}^{(\boldsymbol{\beta},\boldsymbol{\gamma})}\big(\boldsymbol
\chi(1)\big)\leq
\min_{\pi'}J_{\pi'}^{(\boldsymbol{\beta},\boldsymbol{\gamma})}\big(\boldsymbol
\chi(1)\big)$, where $\pi=\{\Omega^1,\Omega^2,\cdots\}$ with
$\Omega^t$ satisfying the structure in Lemma
\ref{Lem:opt-re-ac-pwr-control}. 

\section*{Appendix B: Proof of Lemma \ref{Lem:reduced-MDP}, Corollary \ref{Cor:Bellman-Q-factor} and Corollary \ref{Cor:Bellman-postd}}
\begin{proof} [Lemma \ref{Lem:reduced-MDP}]
Based on Definition \ref{defn:conditional-action1}, we can transform
the POMDP into the MDP with a tuple of the following four objects:
state space $\boldsymbol{\mathcal E}\times \boldsymbol{\mathcal Q}$,
action space $\boldsymbol{\mathcal P} \times \boldsymbol{\mathcal
S}$ with partitioned architecture $\{\Omega^t(\mathbf E,\mathbf
Q)\}$ according to Definition \ref{defn:conditional-action1},
transition kernel $ \Pr[(\mathbf E',\mathbf Q')| (\mathbf E,\mathbf
Q), \Omega(\mathbf E,\mathbf Q)]$, per-stage cost function  $g
\big((\mathbf E,\mathbf Q), \Omega(\mathbf E,\mathbf
Q)\big)$. Since the Weak Accessibility
(WA) condition holds under our problem setup, by Proposition 4.2.3.
in \cite{Bertsekas:2007}, the optimal average cost of the
transformed MDP is the same for all initial states. In addition, by
Proposition 4.1.3. and Proposition 4.1.4. in \cite{Bertsekas:2007},
we know that the solution $\big(\theta, \{V(\mathbf E,\mathbf
Q)\}\big)$ to the Bellman equation in \eqref{eqn:Bellman3} exists.
By Proposition 4.2.1. in \cite{Bertsekas:2007}, we can complete the
proof.
\end{proof}

\begin{proof} [Corollary \ref{Cor:Bellman-Q-factor}] Define $\mathbb Q (\mathbf E,\mathbf Q, \mathbf p)
\triangleq \min_{\Omega_s(\mathbf E,\mathbf Q)}\Big\{ g
\big((\mathbf E,\mathbf Q), \mathbf p, \Omega_s(\mathbf E,\mathbf
Q)\big) + \sum_{(\mathbf E',\mathbf Q')} $ $\Pr[(\mathbf E',\mathbf
Q')| (\mathbf E,\mathbf Q), \mathbf p, \Omega_s(\mathbf E,\mathbf
Q)]V(\mathbf E',\mathbf Q')\Big\}-\theta$. 
Thus, we have $V(\mathbf E,\mathbf Q)=\min_{\mathbf p \in
\boldsymbol{\mathcal P}} \mathbb Q (\mathbf E,\mathbf Q, \mathbf p)
$. Based on \eqref{eqn:Bellman3}, we can obtain \eqref{eqn:Bellman-Q-factor}, which is in terms
of BS-DTX control Q-factor $\{\mathbb Q (\mathbf E,\mathbf Q,
\mathbf p)\}$. From Lemma \ref{Lem:reduced-MDP}, we have the optimal
BS-DTX control action given by
\eqref{eqn:opt-pattern-selection-Q-factor}.
\end{proof}

\begin{proof} [Corollary \ref{Cor:Bellman-postd}] Based on \eqref{eqn:Bellman3}, we can obtain \eqref{eqn:Bellman-postd}\cite{Thesis:Salodkar}. For
any $(\mathbf E, \mathbf Q)\in \boldsymbol{\mathcal E} \times
\boldsymbol{\mathcal Q}$, as $\mathbf p^*=\Omega_p^*(\mathbf E,
\mathbf Q)$ can by obtained by
\eqref{eqn:opt-pattern-selection-Q-factor}, we can obtain $
\Omega_s^*(\mathbf E, \mathbf Q,\mathbf H)$ by solving the R.H.S. of
\eqref{eqn:Bellman-postd} under $\mathbf p^*$ for any $\mathbf A^E$
and $\mathbf A^Q$ as follows:
\begin{align}
&\min_{\Omega_s(\mathbf E,\mathbf Q)}\Big\{ g \Big((\mathbf
E,\mathbf Q), \big(\mathbf p^*,\Omega_s(\mathbf E,\mathbf
Q)\big)\Big) + \sum_{(\widetilde{\mathbf E}',\widetilde{\mathbf
Q}')} \Pr[(\widetilde{\mathbf E}',\widetilde{\mathbf Q}')| (\mathbf
E,\mathbf Q), \big(\mathbf p^*,\Omega_s(\mathbf E,\mathbf
Q)\big)]U(\widetilde{\mathbf E}',\widetilde{\mathbf
Q}')\Big\}\nonumber\\
\stackrel{(a)}{=}& \min_{\Omega_s(\mathbf E,\mathbf Q)}\Big\{
\sum_{(\widetilde{\mathbf E}',\widetilde{\mathbf Q}')}\mathbb E
\big[\Pr[(\widetilde{\mathbf E}',\widetilde{\mathbf Q}')|\boldsymbol
\chi, \big(\mathbf p^*,\Omega_s(\boldsymbol \chi)\big)
]\big|(\mathbf E,\mathbf Q) \big]U(\widetilde{\mathbf
E}',\widetilde{\mathbf Q}')\Big\},\  \forall (\mathbf E, \mathbf
Q)\in \boldsymbol{\mathcal E} \times
\boldsymbol{\mathcal Q}\nonumber\\
\stackrel{(b)}{\Leftrightarrow}& \min_{\Omega_s(\boldsymbol
\chi)}\Big\{\sum_{(\widetilde{\mathbf E}',\widetilde{\mathbf Q}')}
\Pr[(\widetilde{\mathbf E}',\widetilde{\mathbf Q}')|\boldsymbol
\chi, \big(\mathbf p^*,\Omega_s(\boldsymbol \chi)\big)
]\big|(\mathbf E,\mathbf Q) \big]
U(\widetilde{\mathbf E}',\widetilde{\mathbf Q}')\Big\},\ \forall \boldsymbol \chi \in \boldsymbol{\mathcal X}\nonumber\\
\stackrel{(c)}{=}& \min_{\mathbf s\in \boldsymbol{\mathcal
S}(\mathbf p^*)}\Big\{\sum_{\mathbf d\in \boldsymbol{\mathcal
D}}\Big( \prod_{k\in \mathcal K}\big(1-d_k-(-1)^{d_k}\Pr[\rho_k
(\mathbf H, \mathbf p^*, \mathbf s)\geq \delta_k]\big)
U([\mathbf E-\mathbf p^*]^+,[\mathbf Q-\mathbf
d]^+)\Big)\Big\}\label{eqn:b-result}
\end{align}
where (a) is due to the definition of $g(\cdot,\cdot)$ and
$\Pr[(\widetilde{\mathbf E}',\widetilde{\mathbf Q}')| (\mathbf
E,\mathbf Q),\Omega(\mathbf E,\mathbf Q)\big)]$, (b) is due to
Definition \ref{defn:conditional-action1} and (c) is due to Assumptions \ref{Asump:general_A-Q} and \ref{Asump:general_A-E} as well as $E_b-p_b^{E*}=E_b-p_b^{*}\mathbf I[E_b>0]=[E_b-p_b^*]^+$.
\end{proof}

\section*{Appendix C: Proof of Lemma \ref{Lem:decomposition}}

We shall prove the additive property w.r.t. the potential function. Following the proofs of Corollary \ref{Cor:Bellman-Q-factor} and Corollary \ref{Cor:Bellman-postd}, the additive property  can be easily extended to the Q-factor and the post-decision potential function. 
Let $\hat \theta $ and $\hat V(\mathbf E, \mathbf Q)$ be the average cost and the potential function under $\hat \Omega$. Then, we have the following Bellman equation in terms of $(\hat \theta ,\{\hat V(\mathbf E, \mathbf Q)\})$:
\begin{align}
&\hat\theta+\hat V(\mathbf E, \mathbf
Q)=\mathbb E^{\hat \Omega_p}\left[\hat g \big((\mathbf
E,\mathbf Q), \mathbf p \big)\right ]+
\sum_{(\mathbf E',\mathbf Q')}
\mathbb E^{\hat \Omega_p}\left[\hat \Pr\left[(\mathbf E',\mathbf Q')| (\mathbf
E,\mathbf Q), \mathbf p\right]\right]\hat V(\mathbf E',\mathbf Q')\label{eqn:proof:hat-V}
\end{align}
where $\hat g \big((\mathbf
E,\mathbf Q), \mathbf p \big)=g \big((\mathbf
E,\mathbf Q), \mathbf p \big)$ and 
$ \hat \Pr\left[(\mathbf E',\mathbf Q')| (\mathbf
E,\mathbf Q), \mathbf p\right]=\mathbb E\left[\mathbb E^{\hat \Omega_s}\left[\Pr[(\mathbf E', \mathbf Q')|(\mathbf E,\mathbf Q,\mathbf H), \mathbf p, \mathbf s]|\mathbf p,\mathbf H\right]|\mathbf p\right]$.  
Let $\hat \theta_k $ and $\hat V_k(\mathbf E, \mathbf Q)$ be the per-flow average cost and  potential function under $\hat \Omega$. Then, we have the following per-flow fixed point equation in terms of $(\hat \theta_k ,\{\hat V_k(E_b, Q_k)\})$:
\begin{align}
\hat\theta_k+\hat V_k(E_b, Q_k)=\mathbb E^{\hat \Omega_p}\left[\hat g_k \big((
E_b,Q_k), p_b \big)\right ]+
\sum_{(E_b',Q_k')}
\mathbb E^{\hat \Omega_p}\left[\hat \Pr\left[( E_b', Q_k')| (
E_b,Q_k), \mathbf p\right]\right]\hat V_k(E_b',Q_k')\label{eqn:proof:hat-per-user-V}
\end{align}
Under $\hat \Omega$,  the induced Markov chain has a single recurrent class. Therefore, the solutions to \eqref{eqn:proof:hat-V} and \eqref{eqn:proof:hat-per-user-V} exist, respectively. First,  we have $\mathbb E^{\hat \Omega_p}\left[\hat g \big((\mathbf
E,\mathbf Q), \mathbf p \big)\right ]=\sum_{b\in \mathcal B}\sum_{k\in \mathcal K_b}\mathbb E^{\hat \Omega_p}\left[\hat g_k \big((E_b,Q_k), p_b \big)\right ]$.  Second, by the relationship between the joint distribution and the marginal distribution, we have $\sum_{(\mathbf E',\mathbf Q')}
\hat \Pr\left[(\mathbf E',\mathbf Q')| (\mathbf
E,\mathbf Q), \mathbf p\right]=\sum_{( E_b',Q_k')}
\hat \Pr\left[( E_b',Q_k')| (\mathbf
E,\mathbf Q), \mathbf p\right]$ $
=\sum_{( E_b',Q_k')}
\hat \Pr\left[( E_b',Q_k')| ( E_b,Q_k), \mathbf p\right]$. 
Therefore, substitute $\hat \theta=\sum_{k\in \mathcal K}\hat \theta_k$ and $\hat V(\mathbf E, \mathbf Q)=\sum_{b\in \mathcal B}\sum_{k\in \mathcal K_b}$ $\hat V_k(E_b,Q_k)$ into \eqref{eqn:proof:hat-V}, we can see that the equality holds. Therefore, we complete the proof.

\section*{Appendix D: Proof of Lemma \ref{Lem:decentralized-user-selection}}

Using the approximation in \eqref{eqn:approximate-V} and
\eqref{eqn:b-result}, we have
\begin{align}
& \min_{\mathbf s\in \boldsymbol{\mathcal S}(\hat{\mathbf
p}^*)}\Big\{\sum_{\mathbf d\in \boldsymbol{\mathcal D}}\Big(
\prod_{k\in \mathcal K}\big(1-d_k-(-1)^{d_k}\Pr[ \rho_k
\geq \delta_k]\big)\big(\sum_n \sum_{k\in \mathcal K_n} \hat
U_k([E_b-\hat p_b^*]^+,
[Q_k-d_k]^+)\big)\Big)\Big\}\nonumber\\
=&\min_{\mathbf s\in \boldsymbol{\mathcal S}(\hat{\mathbf
p}^*)}\Big\{\sum_n \sum_{k\in \mathcal K_n}\sum_{d_k\in \mathcal D}
\big(1-d_k-(-1)^{d_k}\Pr[\rho_k \geq \delta_k]\big)
\hat U_k([E_b-\hat p_b^*]^+, [Q_k-d_k]^+)\Big\}\nonumber\\
\Leftrightarrow & \min_{\mathbf s_n\in \boldsymbol{\mathcal
S_n}(\hat{\mathbf p}^*_n)} \sum_{k \in \mathcal K_n}\sum_{d_k\in \mathcal
D_k} \big(1-d_k-(-1)^{d_k}\Pr[\rho_k \geq
\delta_k]\big)\hat U_k([E_b-\hat p_b^*]^+, [Q_k-d_k]^+),\ \forall n\nonumber\\
= & \min_{\mathbf s_n\in \boldsymbol{\mathcal S_n}(\hat{\mathbf p}^*_n)}
\sum_{k \in \mathcal K_n}\Big((1-s_k)\hat U_k([E_b-\hat p_b^*]^+, Q_k)\nonumber\\
&\hspace{20mm}+s_k\big(\Pr[\rho_k \geq
\delta_k]\hat U_k([E_b-\hat p_b^*]^+,
[Q_k-1]^+)+(1-\Pr[\rho_k \geq \delta_k]))\hat
U_k([E_b-\hat p_b^*]^+, Q_k)\Big)\nonumber\\
\Leftrightarrow & \min_{\mathbf s_n\in \boldsymbol{\mathcal
S_n}(\hat{\mathbf p}^*_n)} \sum_{k \in \mathcal K_n} s_k\Pr[
\rho_k \geq \delta_k]\big(\hat U_k([E_b-\hat p_b^*]^+, [Q_k-1]^+)-\hat U_k([E_b-\hat p_b^*]^+,
Q_k)\big),\forall n \nonumber
\end{align}



\section*{Appendix E: Proof of Theorem \ref{Thm:performance-comp}}


Under the assumptions  \ref{Asump:general_A-Q} and \ref{Asump:general_A-E} as well as  $\hat \Omega$  in Definition \ref{Def:QSI-ESI-indep-policy}, Markov chain $\{\left(\mathbf E(t),\mathbf Q(t)\right)\}$ has a single recurrent class (and possibly some transient states). Thus,  $\hat \Omega$ is a unchain policy.  In addition, it is obvious that $\hat \Omega^*\neq \hat \Omega$.  Therefore, the conditions of Proposition 4.4.2 in \cite{Bertsekas:2007} are satisfied expect for the assumption that $\hat \Omega^*$ is a unchain policy. We shall modify the proof of Proposition 4.4.2 to incorporate a general $\hat \Omega^*$ as follows. We adopt the same notations as Proposition 4.4.2. ($\mu$ can be treated as $\hat \Omega$ and $\bar \mu$ can be treated as $\hat \Omega^*$). Let $(\bar{\boldsymbol \lambda},\mathbf h_{\bar \mu})$ be the gain-bias pair of a general $\bar \mu$. Thus, by Proposition 4.1.9, $(\bar{\boldsymbol \lambda},\mathbf h_{\bar \mu})$ satisfies $\bar{\boldsymbol \lambda}=\bar P \bar{\boldsymbol \lambda}$ and $\bar{\boldsymbol \lambda}+\mathbf h_{\bar \mu}=\mathbf T_{\bar \mu}\mathbf h_{\bar \mu}$.  However, let $(\lambda,\mathbf h_{\mu})$ be the gain-bias pair of a unchain $\mu$, which satisfies $ \lambda \mathbf e+\mathbf h_{ \mu}=\mathbf T_{\mu}\mathbf h_{\mu}$.  Since $\Pr[(\mathbf E',\mathbf Q')| (\mathbf E,\mathbf Q), (\mathbf p,\mathbf s) ]\neq \Pr[(\mathbf E',\mathbf Q')| (\mathbf E,\mathbf Q), (\mathbf p',\mathbf s')]$, there is strict performance improvement  under $\hat \Omega^*$ over $\hat \Omega$. Thus, we have a stronger result than (4.97), i.e. $\delta (i)>0$ $\forall i$.  To incorporate a general $\bar \mu$, we have $\boldsymbol \delta = (\lambda \mathbf e-\bar{\boldsymbol \lambda})+(\mathbf I-\bar {\mathbf P})\boldsymbol \Delta$ instead of (4.98). Since $\bar{\boldsymbol \lambda}=\bar P \bar{\boldsymbol \lambda}$, we have $\sum_{k=0}^{N-1}\bar{\mathbf P}^k\boldsymbol \delta=N(\lambda \mathbf e-\bar{\boldsymbol \lambda})+(\mathbf I-\bar {\mathbf P}^N)\boldsymbol \Delta$ in stead of (4.99), which implies $\bar{\mathbf P}^*\boldsymbol \delta=\lambda \mathbf e-\bar{\boldsymbol \lambda}$ instead of (4.100). Since $\delta (i)>0$ $\forall i$, we have $\lambda >\bar{\boldsymbol \lambda}(i)$ $\forall i$. In other words, we can show $\hat \theta^*(\mathbf E, \mathbf Q)<\hat \theta$ for all $(\mathbf E, \mathbf Q)\in \boldsymbol{\mathcal E}\times\boldsymbol{\mathcal Q}$.

\section*{Appendix F: Proof of Lemma \ref{Lem:convergence-update}}
Note that the update equations in \eqref{eqn:update-Q} and \eqref{eqn:update-postd-V} can be treated as the synchronous stochastic  versions of the synchronous relative value iterations (RVI) \cite{Bertsekas:2007} for  the Markov chains  $\{\left(E_b(t), Q_k(t), \mathbf p(t) \right )\}$ ($\{\left(E_b(t), Q_k(t) \right)\}$)with  the policy space containing only one policy $\hat \Omega$  \cite{Bertsekas:2007}.
Under $\hat \Omega$ defined in Definition \ref{Def:QSI-ESI-indep-policy}, the two Markov chains have a single recurrent class (and possibly some transient states). Therefore, the condition of Lemma 2 in \cite{YCuiULOFDMA:2010} holds, 
according to the explanation for the conditions of Proposition 4.3.2 in \cite{Bertsekas:2007}. Following the proof of Lemma 2 in \cite{YCuiULOFDMA:2010}, which is a modified version of the proof for Proposition 4.3.2 in \cite{Bertsekas:2007}, we can prove  Lemma \ref{Lem:convergence-update}. 
We omit the details here due to page limit.

\section*{Appendix G: Proof of Lemma \ref{Lem:stability-stability-no-cooperation}}

From the conditional coverage probability (conditioned on the nearest BS being at a distance $r_1$ from the randomly chosen MS) for cellular networks without BS coordination obtained in\cite{JeffSGCellularNet10}, we have  the conditional successful packets transmission probability of the randomly chosen MS given by $p_s(r_1, \lambda')\leq \exp\left(-C_1\lambda' -\frac{N_0}{P}\delta r_1^{\alpha}\right)$, where $\lambda'$ is the
density of the homogeneous PPP $\Phi'$ used to model the locations of
active BSs and the inequality is due to $\eta(x,\alpha)=\int_{x^{-\frac{2}{\alpha}}}^{\infty}\frac{1}{1+u^{\frac{\alpha}{2}}}du\leq \int_{x^{-\frac{2}{\alpha}}}^{\infty}u^{-\frac{\alpha}{2}}du=\frac{2x^{1-\frac{2}{\alpha}}}{\alpha-2}\triangleq \tilde\eta(x,\alpha)$.

Next, we shall show sufficiency by proving that
\eqref{eqn:stability-no-cooperation} guarantees stability in a
parallel dominant network, in which dummy packets are transmitted
when a data queue is empty. Sending dummy packets is only aimed to
cause interference to the other MSs and not counted as an actual
packet transmission. The dominant system stochastically dominates
the original system in the sense that the queue sizes and \textcolor{black}{grid power} costs in that system
are necessarily not smaller (bigger) than those in the original system.
Therefore, the stability conditions obtained for the dominant
systems are sufficient for the stability of the original system. In the dominant system, since $\Pr[p_b=1]=p_{tx}$, we have $\lambda'=p_{tx}\lambda $. 
Therefore, the service rate of the randomly chosen MS is $\mu(p_{tx}, \lambda)=p_{tx}p_s(r_1, \lambda')$. By Loynes' Theorem, the
queue of the randomly chosen MS is stable if $\lambda^Q<\mu(p_{tx}, \lambda)$. Thus, we complete the proof for \eqref{eqn:stability-no-cooperation}.
Note that $E_b$ is decoupled from $Q_k$  and forms a discrete-time
$M/M/1/N_E$ system with arrival rate $\lambda^E$ and departure rate
$p_{tx}$. By queueing theory, we have $\Pr[E_b>0]=f(p_{tx},N_E)$ \cite{EphremidesEH11}.  Thus, we can prove the average \textcolor{black}{grid power} cost in the dominant system is $\overline{p^{G}_{\max}}(p_{tx},N_E)=\big(1-f(p_{tx},N_E)\big)p_{tx}$.


\section*{Appendix H: Proof of Lemma \ref{Lem:stability-stability-cooperation}}

In the following proof, we shall focus on the derivation of the conditional  successful packet transmission probability $p_s(r_1,\lambda,\lambda', N_t)$. The remaining proof is similar to that in the proof of Lemma \ref{Lem:stability-stability-no-cooperation}. Let $b_i$ denote the $i$-th nearest BS among all the BSs (including those are on and off) to the randomly chosen MS $k_0$, where $i=1,\cdots, N_t$. Thus, $b_1$ is the BS of MS $k_0$. By forming a cluster $\mathcal B_0=\{b_1,\cdots, b_{N_t}\}\subset \Phi$, MS $k_0$ can achieve the highest $p_s(r_1,\lambda,\lambda', N_t)$. We shall calculate $p_s(r_1,\lambda,\lambda', N_t)$ under the favorable cluster $\mathcal B_0\subset \Phi$.  Let $R_1$ and $R_{N_t}$ denote the distance between  BS $b_1$ and MS $k_0$ as well as the distance between BS $b_{N_t}$ and MS $k_0$. First, we shall derive the conditional p.d.f. $f_{R_{N_t}|R_1}(r_{N_t}|r_1)$ and the conditional expectation $E\left[R_{N_t}^{2-\alpha}|R_1=r_1\right]$.  If $r_B\leq r_1$, we have $\Pr[R_{N_t}>r_{N_t}|R_1=r_1]=1\Rightarrow f_{R_{N_t}|R_1}(r_{N_t}|r_1)=0 $. It remains to consider $r_{N_t}>r_1$. Let $\mathcal B_2(0,r)$ denote the 2-dim ball centered in the origin with radius $r$.  Following similar techniques in \cite{HaenggiPPPDistanceDist:2005}, we have, for $r_{N_t}>r_1$, 
\begin{align}
&y=\Pr[R_{N_t}>r_{N_t}|R_1=r_1]=\Pr[0, 1,\cdots, N_t-2 \ \text{BSs in}\ \mathcal B_2(0,r_{N_t})-\mathcal B_2(0, r_1) ]\nonumber\\
=&\sum_{i=0}^{N_t-2} \frac{\left( \lambda \pi (r_{N_t}^2-r_1^2)\right)^i}{i!}\exp\left(-\lambda \pi (r_{N_t}^2-r_1^2)\right)\nonumber\\
\Rightarrow& f_{R_{N_t}|R_1}(r_{N_t}|r_1)=-\frac{d y}{d r_{N_t}}=2\lambda \pi r_b \exp\left(-\lambda \pi (r_{N_t}^2-r_1^2)\right) \sum_{i=0}^{N_t-2} \frac{\left( \lambda \pi (r_{N_t}^2-r_1^2)\right)^i}{i!}\nonumber\\
&\hspace{45mm}-\exp\left(-\lambda \pi (r_{N_t}^2-r_1^2)\right) \sum_{i=1}^{N_t-2} \frac{\lambda \pi i\left( \lambda \pi (r_{N_t}^2-r_1^2)\right)^{i-1}2 r_{N_t}}{i!}\nonumber\\
=&2\lambda \pi r_{N_t} \exp\left(-\lambda \pi (r_{N_t}^2-r_1^2)\right)  \frac{\left( \lambda \pi (r_{N_t}^2-r_1^2)\right)^{N_t-2}}{(N_t-2)!}, \quad r_{N_t}>r_1\label{eqn:proof-cond-pdf-rB}\\
\Rightarrow& \mathbb E[R_{N_t}^{2-\alpha}|R_1=r1]=\int_{0}^{\infty} r_{N_t}^{2-\alpha} f_{R_{N_t}|R_1}(r_{N_t}|r_1)d r_{N_t}\nonumber\\
=&\frac{\left( \lambda \pi \right)^{N_t-1}}{(N_t-2)!}\int_{r_1}^{\infty}  (r_{N_t}^2)^{1-\frac{\alpha}{2}} \exp\left(-\lambda \pi (r_{N_t}^2-r_1^2)\right)  (r_{N_t}^2-r_1^2)^{N_t-2}d( r_{N_t}^2-r_1^2)\nonumber\\
\stackrel{(a)}{\leq}& \frac{\left( \lambda \pi \right)^{N_t-1}}{(N_t-2)!}\int_{0}^{\infty} u^{1-\frac{\alpha}{2}}  \exp\left(-\lambda \pi u\right) u^{N_t-2}d u=(-\lambda \pi)^{\frac{\alpha}{2}-1}\frac{\Gamma(N_t-\frac{\alpha}{2})}{{\Gamma(N_t-1)}}, \quad N_t>\frac{\alpha}{2}\label{eqn:proof-cond-E}
\end{align}
 where (a) is due to $\alpha >2$ and the change of variables $u=r_{N_t}^2-r_1^2$. (a) is tight for small $r_1$. In addition, $\mathbb E[R_{N_t}^{2-\alpha}|R_1=r1]\leq r_1^{2-\alpha}$.

Next, we shall calculate $p_s(r_1,\lambda,\lambda', N_t)$. Note that the interference to MS $k_0$ comes from the active BSs in $\Phi'-\mathcal B_0\bigcap \Phi'$. In addition, the signal power $G_1$ and interference power $G_b$ (from the active BS $b\in \Phi'-\mathcal B_0\bigcap \Phi'$) due to small scale fading are exponentially distributed with mean 1\cite{JindalMIMOSISO:2011}. Let $I_{R_{N_t}}$ denote the interference, which is a function of random variable $R_B$. Therefore, we have
\begin{align}
&p_s(r_1,\lambda,\lambda', N_t)=\Pr[SINR\geq \delta|R_1=r_1]=\Pr[\frac{PG_1r_1^{-\alpha}}{N_0+I_{R_{N_t}}}\geq \delta]\nonumber\\
=&\mathbb E_{R_{N_t}}\left[\mathbb E_{I_{R_{N_t}}}\left[\Pr[G_1\geq \frac{1}{P}\delta r_1^{\alpha}(N_0+I_{R_{N_t}})|I_{R_{N_t}}]\right]\right]=\mathbb E_{R_{N_t}}\left[\mathbb E_{I_{R_{N_t}}}\left[\exp\left(-\frac{1}{P}\delta r_1^{\alpha}(N_0+I_{R_{N_t}})\right)\right]\right]\nonumber\\
=&\exp\left(-\frac{N_0}{P}\delta r_1^{\alpha}\right)\mathbb E_{R_{N_t}}\left[\mathbb E_{I_{R_{N_t}}}\left[\exp\left(-\frac{1}{P}\delta r_1^{\alpha}I_{R_{N_t}}\right)\right]\right]\label{eqn:proof-SINR}
\end{align}
Let $s=-\frac{1}{P}\delta r_1^{\alpha}$ and $R_b$ denote the distance between BS $b\in \Phi'-\mathcal B_0\bigcap \Phi'$ and MS $k_0$, we have
\begin{align}
&\mathbb E_{I_{R_{N_t}}}\left[\exp\left(-sI_{R_{N_t}}\right)\right]=\mathbb E_{\Phi', \{G_b\}}\left[\exp\left(-s\sum_{b\in\Phi'-\mathcal B_0\bigcap \Phi'}PG_b R_b^{-\alpha}\right)\right]\nonumber\\
=&\mathbb E_{\Phi', \{G_b\}}\left[\prod_{b\in \Phi'-\mathcal B_0\bigcap \Phi'}\exp\left(-sPG_b R_b^{-\alpha}\right)\right]\nonumber\\
=&\mathbb E_{\Phi'}\left[\prod_{b\in \Phi'-\mathcal B_0\bigcap \Phi'}\mathbb E_{\{G_b\}}\left[\exp\left(-sPG_b R_b^{-\alpha}\right)\right]\right]=\mathbb E_{\Phi'}\left[\prod_{b\in \Phi'-\mathcal B_0\bigcap \Phi'}\frac{1}{1+sPG_b R_b^{-\alpha}}\right]\nonumber\\
=&\exp\left(-2\lambda' \pi\int_{R_{N_t}}\left(1-\frac{1}{1+sPG_b R_b^{-\alpha}}\right)\right) vdv\stackrel{(a)}{=}\exp\left(-2\lambda' \pi\int_{r_1\frac{R_{N_t}}{r_1}}\frac{1}{1+(\frac{v}{r_1\delta^{1/\alpha}})^{-\alpha}}vdv\right) \nonumber\\
\stackrel{(b)}{=}&\exp\left(-\lambda'\pi r_1^2 \delta^{\frac{2}{\alpha}}\int_{\delta^{-\frac{2}{\alpha}}(\frac{R_{N_t}}{r_1})^2}\frac{1}{1+u^{\frac{\alpha}{2}}}vdv\right) =\exp\left(-\lambda'\pi r_1^2 \delta^{\frac{2}{\alpha}}\eta\left(\delta (\frac{r_1}{R_{N_t}})^{\alpha}, \alpha\right)\right) \nonumber\\
\Rightarrow &\mathbb E_{R_{N_t}}\left[\mathbb E_{I_{R_{N_t}}}\left[\exp\left(-\frac{1}{P}\delta r_1^{\alpha}I_{R_{N_t}}\right)\right]\right]
\stackrel{(c)}{\geq}\exp\left(-\lambda'\pi r_1^2 \delta^{\frac{2}{\alpha}} E_{R_{N_t}}\left[\eta(\delta (\frac{r_1}{R_{N_t}})^{\alpha}, \alpha)\right]\right)\nonumber\\
\stackrel{(d)}{\geq}&\exp\left(-\lambda'\pi r_1^2 \delta^{\frac{2}{\alpha}} E_{R_{N_t}}\left[\tilde\eta\left(\delta (\frac{r_1}{R_{N_t}})^{\alpha}, \alpha\right)\right]\right)
=\exp\left(-\lambda'\frac{2}{\alpha-2}\pi r_1^{2 }\delta r_1^{\alpha-2}E_{R_{N_t}}\left[(R_{N_t})^{2-\alpha}|R_1=r_1\right]\right)\nonumber\\
\stackrel{(e)}{\geq}&\exp\left(-C_{N_t}\lambda' \right)\label{eqn:proof-SIR}
\end{align}
where (a) is due to plugging in $s=-\frac{1}{P}\delta r_1^{\alpha}$, (b) is due to the change of variables $u=(\frac{v}{r_1\delta^{1/\alpha}})^{-\alpha}$, (c) is due to the convexity of the exponential function, (d) is due to $\eta(x,\alpha)=\int_{x^{-\frac{2}{\alpha}}}^{\infty}\frac{1}{1+u^{\frac{\alpha}{2}}}du\leq \int_{x^{-\frac{2}{\alpha}}}^{\infty}u^{-\frac{\alpha}{2}}du=\frac{2x^{1-\frac{2}{\alpha}}}{\alpha-2}\triangleq \tilde\eta(x,\alpha)$, (e) is due to inequality \eqref{eqn:proof-cond-E} and $\mathbb E[R_{N_t}^{2-\alpha}|R_1=r1]\leq r_1^{2-\alpha}$.  Substituting \eqref{eqn:proof-SIR} into \eqref{eqn:proof-SINR}, we have $p_s(r_1,\lambda,\lambda', N_t)\geq \exp\left(-C_{N_t}\lambda'-\frac{N_0}{P}\delta r_1^{\alpha}\right)$. Since $\lambda'=p_{tx}\lambda$, we can prove \eqref{eqn:stability-cooperation}.

\section*{Appendix I: Proof of Theorem \ref{Thm:prop-stability-stability-cooperation}}

\begin{align}
&\frac{d\lambda^Q_{\max}}{dp_{tx}}=
(1-C_{N_t}\lambda
p_{tx}\lambda) \exp\left(-C_{N_t}p_{tx}\lambda -\frac{N_0}{P}\delta r_1^{\alpha}\right)
\begin{cases}
>0, &  p_{tx}<\frac{1}{C_{N_t} \lambda}\\
\leq 0, &  p_{tx}\geq\frac{1}{C_{N_t} \lambda}
\end{cases}\label{eqn:d-y}\\
&\frac{\partial\overline{p^{G}_{\max}}(p_{tx},N_E)}{\partial p_{tx}}=1-\frac{\partial f(p_{tx},N_E)}{\partial p_{tx}}+\left(1-f(p_{tx},N_E)\right)>0, \ \forall N_E>0 \nonumber
\end{align}
where the last inequality is due to $\frac{\partial f(p_{tx},N_E)}{\partial p_{tx}}<0$ and $f(p_{tx},N_E)\leq1$. In addition, $p_{tx}\in [0,1]$. Therefore, we can easily obtain $p_{tx}^*(N_E,N_t)$. Next, we shall prove the property of $\lambda^{Q*}_{\max}(N_E,N_t)$ w.r.t. $N_E$. It is obvious that $x^*(N_E)$ increases with $N_E$. If $x^*(N_E)<\min\{1,\frac{1}{C_{N_t}\lambda}\}$ , $p_{tx}^*(N_E)=x^*(N_E)$ and $\frac{d\lambda^Q_{\max}}{dp_{tx}}\big|_{p_{tx}=p_{tx}^*(N_E,N_t)}>0$. Thus, $\lambda^{Q*}_{\max}(N_E,N_t)$ is increasing in $N_E$. If $x^*(N_E)\geq \min\{1,\frac{1}{C_{N_t} \lambda}\}$, $p_{tx}^*(N_E)=\min\{1,\frac{1}{C_{N_t} \lambda}\}$. $\lambda^{Q*}_{\max}(N_E,N_t)$ is a constant for all $N_E$. Finally, we shall show  the property of $\lambda^{Q*}_{\max}(N_E,N_t)$ w.r.t. $N_t$ for any given $N_E>0$. It can be easily verified that $C_{N_t}$ is decreasing in $N_t$.  Thus, when $p_{tx}^*(N_E,N_t)=x^*(N_E)$ or $1$, we have that $\lambda^{Q*}_{\max}(N_E,N_t)=p_{tx}^*(N_E) \exp\left(-C_{N_t}p_{tx}^*(N_E)\lambda -\frac{N_0}{P}\delta r_1^{\alpha}\right)$ is increasing in $N_t$. When $p_{tx}^*(N_E,N_t)=\frac{1}{C_{N_t} \lambda}$, we have that $\lambda^{Q*}_{\max}(N_E,N_t)=C_{N_t} \lambda \exp\left(-1 -\frac{N_0}{P}\delta r_1^{\alpha}\right)$ is increasing in $N_t$.

\end{appendix}


\begin{figure}[h]
\begin{center}
  \subfigure[\textcolor{black}{System Model}]
  {\resizebox*{5cm}{8cm}{\includegraphics{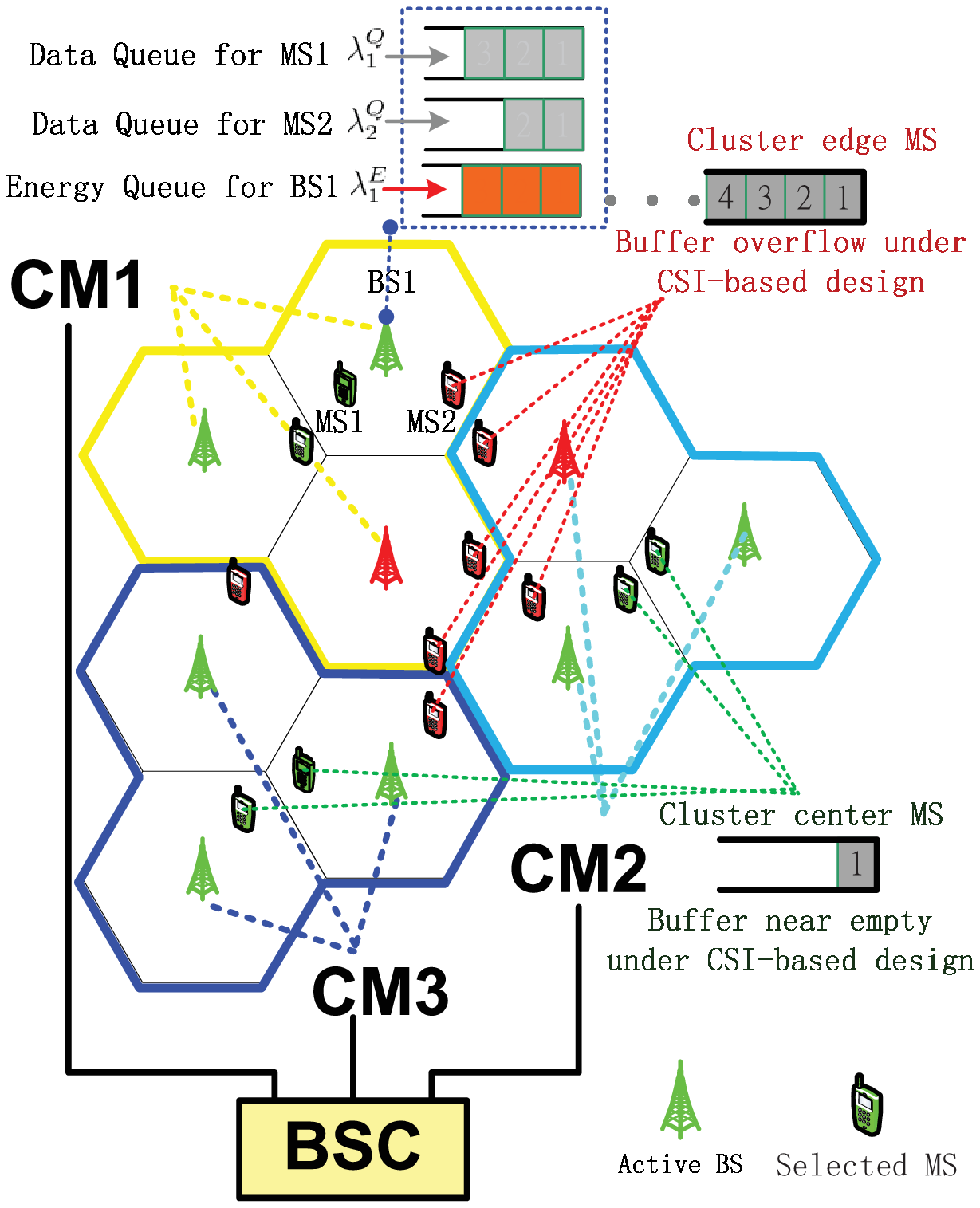}}}
  \subfigure[Control Architecture]
  {\resizebox*{10cm}{7cm}{\includegraphics{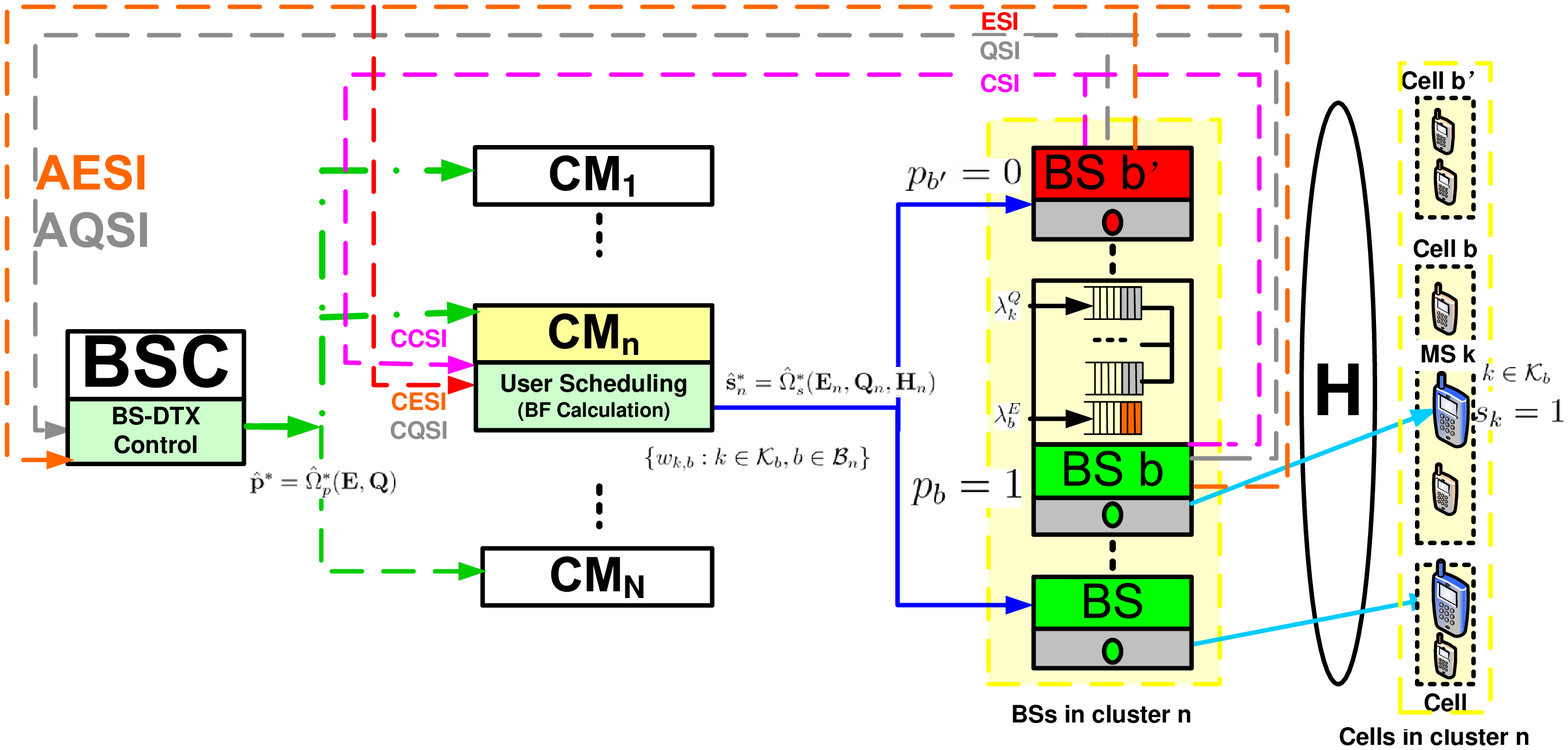}}}
  \end{center}
  \caption{System model and control architecture of the downlink coordinated MIMO  systems. The dotted lines and solid
lines on Fig \ref{fig:system-model}. (b) denote the control path and
data path, respectively.}
  \label{fig:system-model}
\end{figure}

\begin{figure}[t]
\begin{center}
\includegraphics[height=8cm, width=10cm]{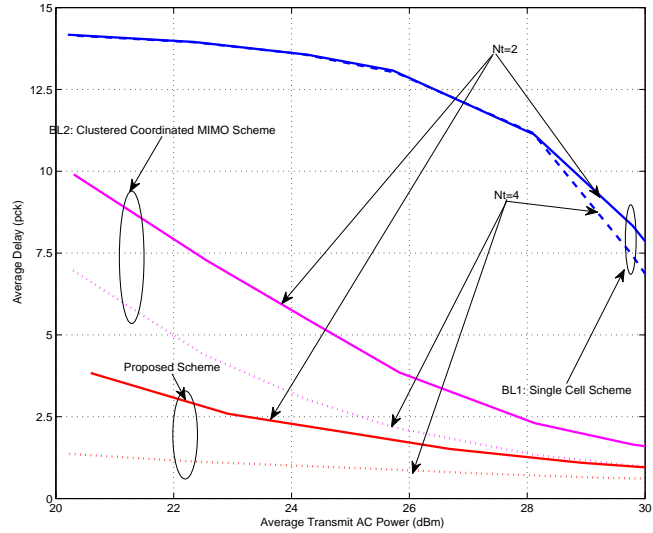}
\caption{Average delay versus average transmit \textcolor{black}{grid power} (dBm) at $N_t=2$ and $N_t=4$. The average data
arrival rate is $\lambda^Q_k=0.4$ pck/slot, the average renewable energy arrival rate is $\lambda^E_b=0.5$ unit/slot,  and the energy storage size $N_E=4$ units for all $k\in \mathcal K$ and $b\in \mathcal B$.}
\label{Fig:delay_vs_pwr}
\end{center}
\end{figure}

\begin{figure}[t]
\begin{center}
\includegraphics[height=8cm, width=10cm]{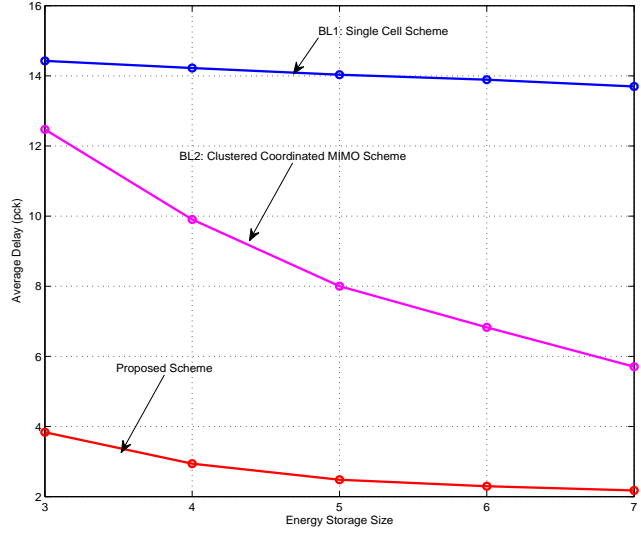}
\caption{Average delay versus energy storage size at average transmit \textcolor{black}{grid power} 15 dBm. The average data
arrival rate is $\lambda^Q_k=0.4$ pck/slot and the average renewable energy arrival rate is $\lambda^E_b=0.6$ unit/slot  for all $k\in \mathcal K$ and $b\in \mathcal B$.}
\label{Fig:delay_vs_NE}
\end{center}
\end{figure}

\begin{figure}[t]
\begin{center}
\includegraphics[height=8cm, width=10cm]{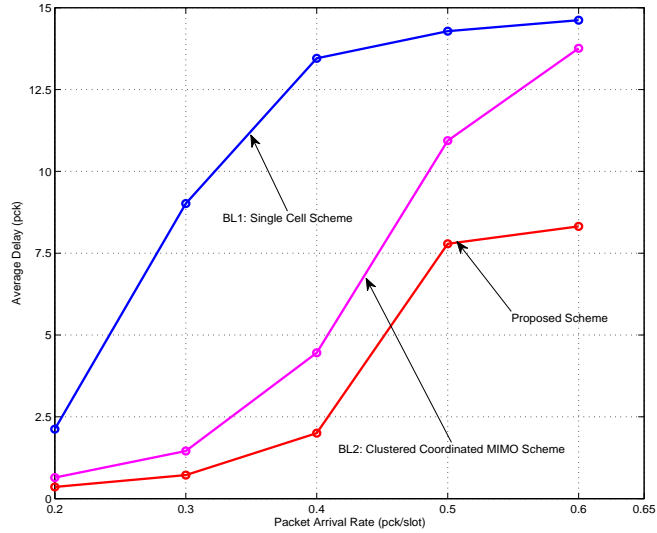}
\caption{Average delay versus average data arrival rate at average transmit \textcolor{black}{grid power} 25 dBm. The average renewable energy arrival rate is $\lambda^E_b=0.5$ unit/slot and the energy storage size $N_E=4$ for all $b\in \mathcal B$.}
\label{Fig:delay_vs_arrival}
\end{center}
\end{figure}

\begin{figure}[t]
\begin{center}
\includegraphics[height=8cm, width=10cm]{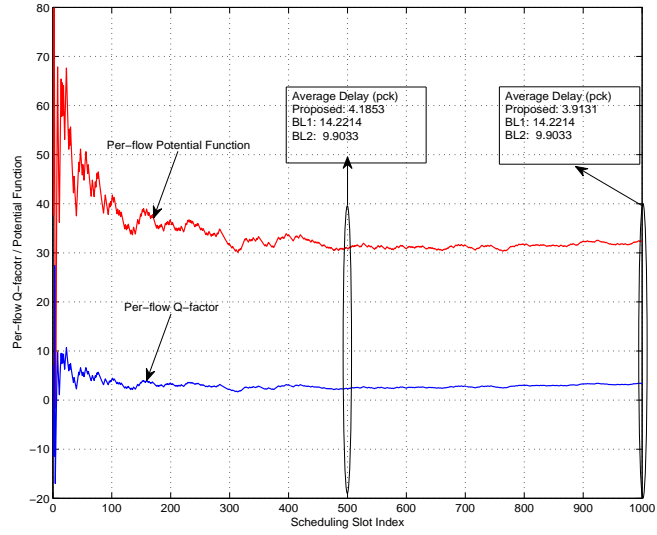}
\caption{Convergence property of the proposed distributed
online learning algorithm at average transmit \textcolor{black}{grid power} 20 dBm. The average data
arrival rate is $\lambda^Q_k=0.4$ pck/slot, the average renewable energy arrival rate is $\lambda^E_b=0.5$ unit/slot,  and the energy storage size $N_E=4$ for all $k\in \mathcal K$ and $b\in \mathcal B$. The figure illustrate the instantaneous
per-flow post-decision potential function value  $\hat U_k^t(\widetilde E_b, \widetilde Q_k)$ and the
instantaneous per-flow Q-factor value $\hat {\mathbb Q}_k^t(E_b, Q_k,\mathbf p)$ respectively (during the
online iterative updates in \eqref{eqn:update-postd-V} and
\eqref{eqn:update-Q}) versus instantaneous slot index, where $k=1$, $b=1$, $E_b=1$, $Q_k=1$, $\widetilde E_b=1$, $\widetilde Q_k=1$ and $\mathbf p=\{p_b=1:b\in \mathcal B\}$. The boxes
indicate the average delay performance of various schemes at the two
selected slot indices.}
\label{Fig:convergence}
\end{center}
\end{figure}

\end{document}